\documentclass{IEEEoj}
\usepackage{cite}
\usepackage{amsmath,amssymb,amsfonts}
\usepackage{algorithmic}
\usepackage{graphicx,color}
\usepackage{textcomp}

\usepackage{braket}
\usepackage{mathtools}
\usepackage{amsmath}
\usepackage{amssymb}
\usepackage{amsfonts}
\usepackage{lipsum}
\usepackage{mathtools}
\usepackage{lscape}
\usepackage{multirow}
\usepackage{here}
\usepackage{bm}
\usepackage{cuted}
\usepackage{adjustbox}
\usepackage{cite}
\usepackage{tabularx}
\let\labelindent\relax
\usepackage{enumitem}
\usepackage{multicol}
\usepackage{booktabs}
\usepackage{makeidx}

\usepackage{setspace}

\usepackage{here}
\usepackage{subfigure}
\usepackage{amsthm}
\newtheorem{definition}{Definition}
\newtheorem{theorem}{Theorem}

\usepackage{array}

\usepackage{dsfont}

\usepackage[ruled,vlined]{algorithm2e}

\SetKwInput{KwInput}{Input}                %
\SetKwInput{KwOutput}{Output}

\newcommand{\mb}[1]{\mathbf{#1}}

\newcommand{\argmax}{\mathop{\rm arg~max}\limits}

\def\BibTeX{{\rm B\kern-.05em{\sc i\kern-.025em b}\kern-.08em
    T\kern-.1667em\lower.7ex\hbox{E}\kern-.125emX}}
\AtBeginDocument{\definecolor{ojcolor}{cmyk}{0.93,0.59,0.15,0.02}}
\def\OJlogo{\vspace{-14pt}\includegraphics[height=28pt]{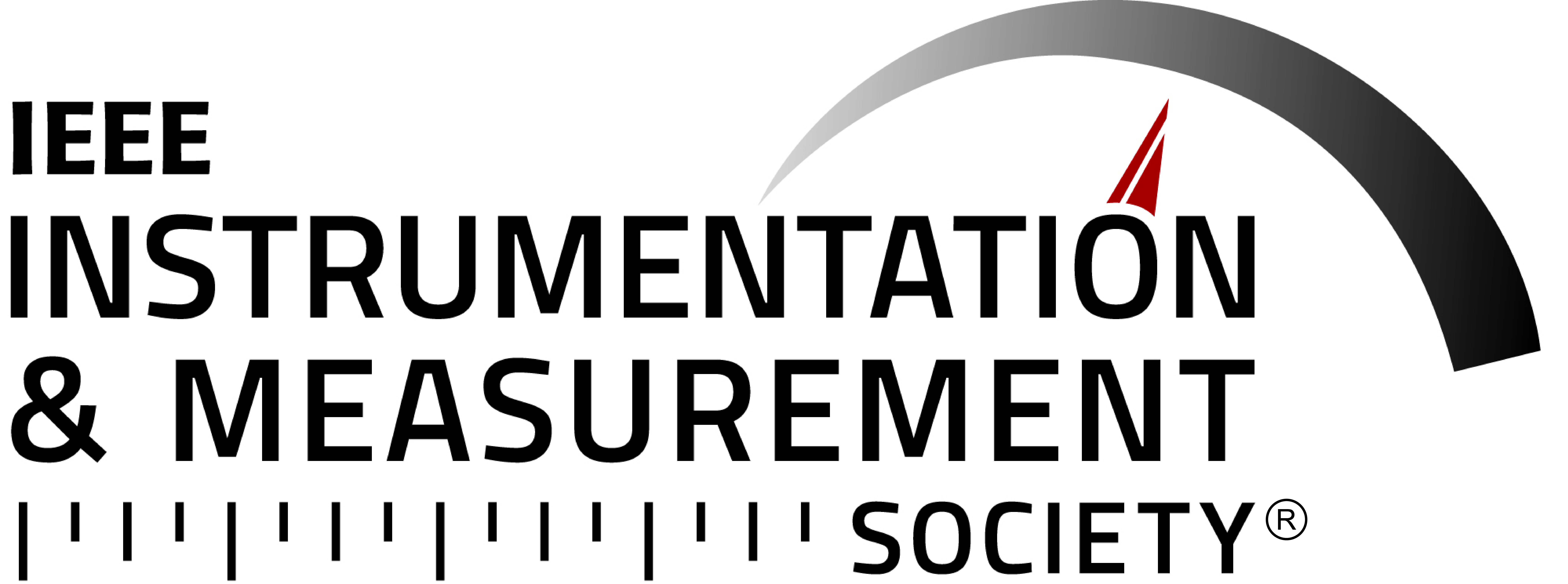}}
\begin{document}
\receiveddate{XX Month, XXXX}
\reviseddate{XX Month, XXXX}
\accepteddate{XX Month, XXXX}
\publisheddate{XX Month, XXXX}
\currentdate{XX Month, XXXX}
\doiinfo{OJIM.2022.1234567}

\title{Multi-channel Sampling on Graphs and \\Its Relationship to Graph Filter Banks}

\author{JUNYA HARA$^1$, YUICHI TANAKA$^{2,3}$}
\affil{Department of Electrical Engineering and Computer Science, Tokyo University of Agriculture and Technology, Tokyo, 184--8588 Japan}
\affil{Graduate School of Engineering, Osaka University, Osaka, 565--0871 Japan}
\affil{PRESTO, Japan Science and Technology Agency, Saitama, 332--0012 Japan}
\corresp{CORRESPONDING AUTHOR: Junya Hara (e-mail: jhara@msp-lab.org).}
\authornote{This work is supported in part by JST PRESTO (JPMJPR1935) and JSPS KAKENHI (20H02145, 22J22176).}
\markboth{MULTI-CHANNEL SAMPLING ON GRAPHS AND ITS RELATIONSHIP TO GRAPH FILTER BANKS}{HARA AND TANAKA}

\begin{abstract}
In this paper, we consider multi-channel sampling (MCS) for graph signals. 
We generally encounter full-band graph signals beyond the bandlimited ones in many applications, such as piecewise constant/smooth graph signals and union of bandlimited graph signals. Full-band graph signals can be represented by a mixture of multiple signals conforming to different generation models. 
This requires the analysis of graph signals via multiple sampling systems, i.e., MCS, while existing approaches only consider single-channel sampling.
We develop a MCS framework based on generalized sampling.
We also present a sampling set selection (SSS) method for the proposed MCS so that the graph signal is best recovered.
Furthermore, we reveal that existing graph filter banks can be viewed as a special case of the proposed MCS.
In signal recovery experiments, the proposed method exhibits the effectiveness of recovery for full-band graph signals.

\end{abstract}

\maketitle

\section{INTRODUCTION}
\label{sec:intro}

\IEEEPARstart{G}{raph} signal processing (GSP) is a fundamental theory for analyzing graph-structured data, i.e., graph signals \cite{ortega_graph_2018}. 
Sampling of graph signals is one of the central research topics in GSP \cite{tanaka_sampling_2020}.
Sampling theory for graph signals, hereafter we call it \textit{graph sampling theory}, can be applied to various applications, including sensor placement \cite{sakiyama_eigendecomposition-free_2019}, traffic monitoring \cite{chen_monitoring_2016}, semi-supervised learning \cite{gadde_active_2014,anis_sampling_2019}, and graph filter bank designs \cite{hammond_wavelets_2011,narang_compact_2013,tanaka_m_2014,sakiyama_two-channel_2019,narang_perfect_2012,sakiyama_oversampled_2014,tay_techniques_2015,tay_critically_2017,tay_bipartite_2017,tremblay_subgraph-based_2016,pavez_two_2022}.

Most studies on graph sampling theory focus on the bandlimited graph signal model as an analog of the classical sampling theory for time-domain signals \cite{chen_discrete_2015,tsitsvero_signals_2016,anis_efficient_2016,marques_sampling_2016}.
However, we often encounter full-band graph signals in many applications. For example, piecewise smooth graph signals and multi-band graph signals are classified into full-band signals.
While some works study graph signal sampling beyond the bandlimited model \cite{hara_graph_2022,tanaka_generalized_2020,hara_sampling_2022}, they consider signals under one signal model: Signals with the mixture of two or more signal models cannot be recovered properly.

Generally, full-band signals can be represented by a mixture of multiple signals conforming to different generation models. 
For recovering such signals, we need to consider multiple sampling systems, i.e., \textit{multi-channel sampling} (MCS), where its single-channel sampling and recovery correspond to one signal model. 

For standard time-domain signals, MCS has been studied as the Papoulis' sampling theorem \cite{papoulis_generalized_1977}: The ideally-bandlimited signal can be recovered from samples obtained by MCS  with arbitrary $M$ sampling methods, 
e.g., non-uniform sampling and bandpass sampling. 
Later, it was extended into the full-band case \cite{djbkovic_generalized_1997,christensen_generalized_2005,aldroubi_oblique_1996}. 
This can be viewed as a special case of generalized sampling \cite{eldar_sampling_2015}.
It is composed of sampling, correction, and reconstruction transforms. Sampling and reconstruction transforms can be arbitrarily chosen while the correction transform compensates for their non-ideal behaviors, ensuring that the reconstructed signal is in some sense close to the original signal.
From a generalized sampling perspective, MCS can be viewed as one of the possible sampling transforms.
While numerous works have been presented for time-domain MCS \cite{brown_multi-channel_1981,eldar_filterbank_2000,djbkovic_generalized_1997,christensen_generalized_2005,aldroubi_oblique_1996},
there has been no approach to MCS in the graph setting in spite of having various full-band graph signals in many applications. 

In this paper, 
we consider MCS for GSP to recover full-band graph signals.
The proposed MCS is derived from the above-mentioned generalized sampling by extending the sampling transform for the graph setting \cite{eldar_sampling_2015}.
We also design the sampling transform for MCS on graphs. It requires the selection of a subset of vertices, i.e., sampling set selection (SSS). We select the sampling set such that graph signals are best recovered.
One can notice that MCS is related to filter banks.
In fact, sampling of full-band graph signals has been studied in a different line of research: Graph filter bank (GFB) designs \cite{hammond_wavelets_2011,narang_compact_2013,tanaka_m_2014,sakiyama_two-channel_2019,narang_perfect_2012,sakiyama_oversampled_2014,tay_techniques_2015,tay_critically_2017,tay_bipartite_2017,tremblay_subgraph-based_2016,pavez_two_2022}.
GFBs are composed of multiple (typically low- and high-pass) graph filters and down- and up-sampling operators, which are also components in MCS. 
Typically, perfect reconstruction (PR) GFBs are designed based on the properties of the given graph operator (e.g., adjacency matrix or graph Laplacian).

Bipartite graph filter banks (BGFB) are one of the PR GFBs and they are designed so that graph signals on bipartite graphs are perfectly recovered \cite{narang_compact_2013,narang_perfect_2012,sakiyama_oversampled_2014,tay_techniques_2015,tay_critically_2017,tay_bipartite_2017}.
While BGFBs can satisfy several desirable properties of GFB, they have two major limitations: 1) their PR property is limited to signals on the bipartite graph, and 2) 
BGFBs as well as many GFBs require the eigendecomposition of the graph operator to implement analysis and synthesis filters \cite{narang_perfect_2012,sakiyama_oversampled_2014,tay_techniques_2015,tay_critically_2017,tay_bipartite_2017}\footnote{There are few exceptions like methods in \cite{narang_compact_2013,tremblay_subgraph-based_2016}, but they typically need careful filter designs.}, which is computationally expensive for large graphs.

The proposed method overcomes the limitations of BGFBs by viewing GFBs as a special case of MCS for graph signals:
MCS can guarantee PR for arbitrary graph signals independent of the graph operator. Therefore, it does not require the graph simplification (typically bipartition). Furthermore, MCS allows for the use of arbitrary graph filters and down- and up- sampling operators. 
Our MCS can be implemented without eigendecomposition by utilizing polynomial filters for both analysis and synthesis, while many existing methods require eigendecomposition for achieving PR.
Recovery experiments demonstrate that the proposed method outperforms exiting GFBs.

\textit{Notation:} Bold lower and upper cases represent a vector and matrix, respectively.
We denote $\ell_2$ and spectral norms by $\|\cdot\|$ and $\|\cdot\|_2$, respectively. 
$\mb{A}_{\mathcal{XY}}$ and $\mb{A}_{\mathcal{X}}$ denote submatrices of $\mb{A}$ indexed by $\mathcal{X}$ and $\mathcal{Y}$, and $\mathcal{X}$ and $\mathcal{X}$, respectively. $\mb{A}_{xy}$ denotes the $(x,y)$th element of $\mb{A}$. $\mb{A}^\mathsf{T}$ denotes the transpose of $\mb{A}$. $\mathcal{M}^c$ denotes the complement set of $\mathcal{M}$. $|\mathcal{M}|$ represents the cardinality of the set $\mathcal{M}$. We denote the Kronecker delta function centered at the $y$th element by $\bm{\delta}_y$.

We consider a weighted undirected graph $\mathcal{G}=(\mathcal{V,E})$,
 where $\mathcal{V}$ and $\mathcal{E}$ represent sets of vertices and edges, respectively. The number of vertices is $N = |\mathcal{V}|$ unless otherwise specified. The adjacency matrix of $\mathcal{G}$ is denoted by $\mb{W}$ where its $(m,n)$-element $\mb{W}_{mn}\geq 0$ is the edge weight between the $m$th and $n$th vertices; $\mb{W}_{mn}=0$ for unconnected vertices. The degree matrix $\mathbf{D}$ is defined as $\mathbf{D}=\text{diag}\,(d_{0},d_{1},\ldots,d_{N-1})$, where $d_{m}=\sum_n\mb{W}_{mn}$ is the $m$th diagonal element.
We use graph Laplacian $ \mathbf{L}\coloneqq \mathbf{D}-\mb{W}$ as a graph operator. A graph signal $\bm{x} \in \mathbb{R}^N$ is defined as a mapping from the vertex set to the set of real numbers, i.e., $x[n]:\mathcal{V}\rightarrow \mathbb{R}$.

The graph Fourier transform (GFT) of $\bm{x}$ is defined as
$
\hat{\bm{x}}=\mb{U}^\mathsf{T}\bm{x}
$
where the GFT matrix $\mb{U}$ is 
obtained by the eigendecomposition of the graph Laplacian $\mathbf{L}=\mathbf{U\Lambda U}^\mathsf{T}$ with the eigenvalue matrix $\bm{\Lambda}=\mathrm{diag}\,(\lambda_0,\lambda_1,$ $\ldots,\lambda_{N-1})$. We refer to $\lambda_i$ as a \textit{graph frequency}.

\begin{figure}[t!]
\centering
  \includegraphics[width=\columnwidth]{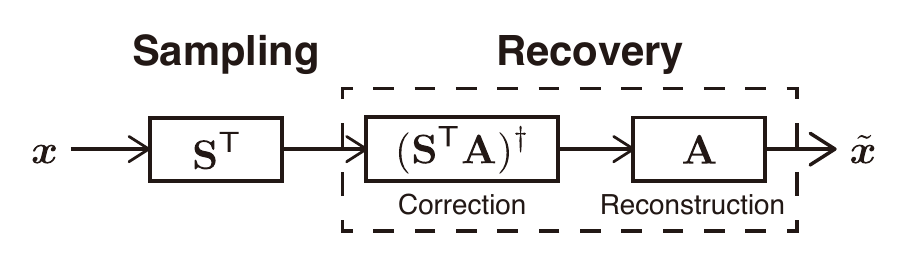}\vspace{-1em}
  \caption{The framework of single-channel sampling.}\label{framework_single}
\end{figure}

\section{SAMPLING FRAMEWORK FOR GRAPH SIGNALS}

In this section, we briefly review preliminary works on single-channel graph signal sampling \cite{hara_sampling_2022}. First, we introduce a generalized sampling framework on graphs \cite{tanaka_sampling_2020}. Second, we introduce the SSS under subspace priors.

\subsection{SAMPLING UNDER SUBSPACE PRIORS}\label{subsec:samp_subspace}

Suppose that graph signals are characterized by the following linear model:
\begin{align}
\bm{x}\coloneqq \mb{A}\bm{d},\label{eq:generator}
\end{align}
where $\mb{A}\in\mathbb{R}^{N\times K}\ (K\leq N)$ is a known generator matrix and $\bm{d}\in\mathbb{R}^{K}$ are expansion coefficients.
The generator matrix $\mb{A}$ specifies the signal subspace $\mathcal{A}$.

Let $\mb{S}^\mathsf{T}\in\mathbb{R}^{K\times N}$ be a sampling operator.
Regardless of the choice of $\mb{S}^\mathsf{T}$ (and $\mb{A}$), the best possible recovery is always given by \cite{eldar_sampling_2015,tanaka_generalized_2020,hara_graph_2022}
\begin{align}
\tilde{\bm{x}}=\mb{A}\mb{H}\bm{c}=\mb{A}(\mb{S}^\mathsf{T}\mb{A})^\dagger\mb{S}^\mathsf{T}\bm{x},\label{eq:recov_subspace}
\end{align}
where $\dagger$ is the Moore-Penrose inverse.
The framework is illustrated in Fig. \ref{framework_single}.
The sampling matrix $\mb{S}^\mathsf{T}$ specifies the sampling subspace $\mathcal{S}$.
The correction matrix is therefore represented as $\mb{H} = (\mb{S}^\mathsf{T}\mb{A})^\dagger$.
If $\mathcal{A}$ and $\mathcal{S}$ together span $\mathbb{R}^N$ and only intersect at the origin, perfect recovery, i.e.,  $\tilde{\bm{x}} = \bm{x}$, is obtained with $(\mb{S}^\mathsf{T}\mb{A})^{-1}$. We refer to this condition as the \textit{direct sum} (DS) condition \cite{tanaka_sampling_2020}.

While this paper focuses on the case that the signal subspace is known, we can recover signals without the exact knowledge of $\mathcal{A}$ with appropriate priors such as smoothness and stochastic priors \cite{tanaka_generalized_2020,hara_graph_2022,eldar_sampling_2015,chepuri_graph_2018}.

\subsection{SAMPLING SET SELECTION FOR FULL-BAND GRAPH SIGNALS}\label{sec:sss_single}

We introduce vertex domain sampling in the single-channel setting. Since there is no regular sampling in the graph setting, it requires the selection of a subset of vertices, which is hereafter referred to as sampling set selection (SSS).
Vertex domain sampling operator is defined as follows:

\begin{definition}[Vertex domain sampling]
Let $\mb{I}_{\mathcal{MV}}\in\{0,1\}^{K\times N}$ be the submatrix of the identity matrix indexed by $\mathcal{M}\subset \mathcal{V}$ $(|\mathcal{M}|=K)$ and $\mathcal{V}$. The sampling operator is defined by
\begin{align}
\mb{S}^\mathsf{T}\coloneqq \mb{I}_{\mathcal{MV}}\mb{G},\label{eq:node_samp_ope_part3}
\end{align}
where $\mb{G}\in\mathbb{R}^{N\times N}$ is an arbitrary graph filter. A sampled graph signal is thus given by $\bm{y}=\mb{S}^\mathsf{T}\bm{x}$.
\end{definition}

\noindent

There exist several approaches of SSS, i.e., the design of $\mb{I}_{\mathcal{MV}}$, however, most methods assume the bandlimited graph signals.
In this paper, we consider sampling of full-band graph signals since the analysis of graph signals beyond the bandlimited assumption is generally necessary in a multi-channel setting.
To this aim, we introduce the quality of sampling based on the DS condition:
We consider the following problem as a sampling set selection:
\begin{align}
\mathcal{M}^*=\mathop{\arg \max}_{\mathcal{M}\subset \mathcal{V}} \det (\mb{Z}_\mathcal{M}),\label{eq:detmax1}
\end{align}
where $\mb{Z}=\mb{G}\mb{A}\mb{A}^\mathsf{T}\mb{G}^\mathsf{T}$. Since $\det(\mb{Z}_\mathcal{M})=\det(\mb{S}^\mathsf{T}\mb{A}\mb{A}^\mathsf{T}\mb{S})=|\det(\mb{S}^\mathsf{T}\mb{A})|^2$, \eqref{eq:detmax1} encourages that the DS condition is satisfied.
The cost function in \eqref{eq:detmax1} is designed based on the D-optimal design, ensuring that the direct sum condition is satisfied. Unlike other SSS methods, (4) is applicable to a wide range of signal models beyond the bandlimited assumption. Further details can be found in \cite{hara_sampling_2022}.

The direct maximization of \eqref{eq:detmax1} is combinatorial and is practically intractable.
Therefore, we apply a greedy method to \eqref{eq:detmax1}.
Suppose that $\text{rank}(\mb{Z})\geq K$. 
By applying the Schur determinant formula \cite{zhang_schur_2006} to \eqref{eq:detmax1}, it results in
\begin{align}
y^*=&\mathop{\arg \max}_{y\in \mathcal{M}^c} \det (\mb{Z}_{\mathcal{M}\cup\{y\}})\nonumber\\
= &\mathop{\arg \max}_{y\in \mathcal{M}^c}\det (\mb{Z}_{\mathcal{M}})\cdot( \mb{Z}_{y,y}-\mb{Z}_{y,\mathcal{M}}(\mb{Z}_{\mathcal{M}})^{-1}\mb{Z}_{\mathcal{M},y})\nonumber\\
= &\mathop{\arg \max}_{y\in \mathcal{M}^c}\mb{Z}_{y,y}-\mb{Z}_{y,\mathcal{M}}(\mb{Z}_{\mathcal{M}})^{-1}\mb{Z}_{\mathcal{M},y},\label{eq:maxdet2}
\end{align}
where we omit the multiplication with $\det (\mb{Z}_{\mathcal{M}})$ in the third equivalence because it does not depend on $y^*$. 

Still, \eqref{eq:maxdet2} is computationally expensive due to the matrix inversion, which typically requires $O(|\mathcal{M}|^3)$ computational complexity.
To alleviate this, we utilize the Neumann series for $(\mb{Z}_\mathcal{M})^{-1}$. We omit the detail due to the limitation of the space. Please refer to \cite{hara_sampling_2022}. As a result, the SSS algorithm is described by Algorithm \ref{algo:prop_single}.

\begin{algorithm}[!h]\label{algo:prop_single}
\DontPrintSemicolon
\setlength{\abovedisplayskip}{0pt}
\setlength{\belowdisplayskip}{0pt}
  \KwInput{$\mb{Z},\mathcal{M}=\emptyset, K, m=0$}
  Set $\alpha_0^*$ such that $\|\mb{I}-\alpha_0^*\mb{Z}_\mathcal{M}\|_2\leq 1$\\
  \While{$|\mathcal{M}|<K$}{
  Compute $\bm{\varepsilon}_y^{0}=\alpha^*_0\mb{Z}_\mathcal{M}\bm{\delta}_y$\\
    \While{$\|\bm{\varepsilon}_y^{0}-\alpha_m^*\mb{Z}_\mathcal{M}\bm{\varepsilon}_y^{m}\|\geq \beta$ for some $\beta>0$}{
      $\alpha_{m+1}^*\leftarrow \frac{(\bm{\varepsilon}^m_y)^\mathsf{T}\mb{Z}_\mathcal{M}\bm{\varepsilon}^m_y}{\|\mb{Z}_\mathcal{M}\bm{\varepsilon}^m_y\|^2}$\\
      $\bm{\varepsilon}^{m+1}_y \leftarrow \bm{\varepsilon}_y^{0}+(\mb{I}-\alpha^*_{m+1}\mb{Z}_\mathcal{M})\bm{\varepsilon}_y^{m}$\\
      $m\leftarrow m+1$\\
     }
    $y^*\rightarrow\mathop{\arg\max}_{y\in\mathcal{M}^c} \mb{Z}_{yy}-\mb{Z}_{y,\mathcal{M}}\bm{\varepsilon}_y^{m}$\\
    $\mathcal{M}\leftarrow \mathcal{M}\cup \{y\}$
  }
  \KwOutput{$\mathcal{M}$}
\caption{Single-channel SSS}
\end{algorithm}

In the following, we extend the single-channel sampling to the multi-channel sampling. The MCS parallels most of the formulation of the single-channel one.

\section{PROPOSED MULTI-CHANNEL SAMPLING ON GRAPHS}

In this section, we build a MCS framework by extending the single channel sampling introduced in the previous section.  The framework is illustrated in Fig. \ref{framework_multi}. In addition, we convert the MCS into an equivalent subband-wise expression. This allows for an efficient computation of the recovery transform. Based on the framework, we develop the SSS for our MCS such that full-band graph signals are best recovered.

\begin{figure}[t!]
\centering
  \includegraphics[width=1.\columnwidth]{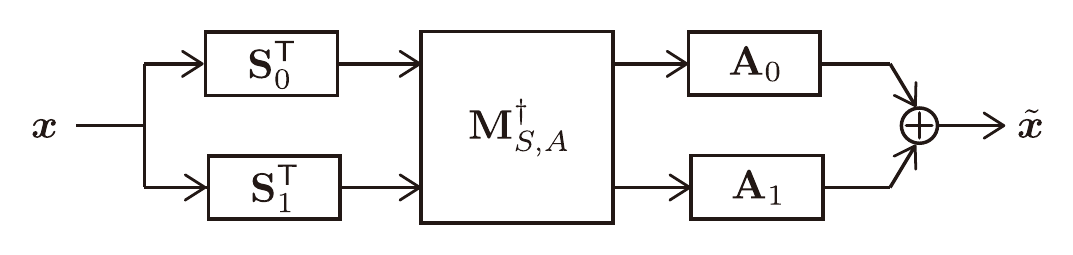}\vspace{-1em}
  \caption{The framework of MCS in the case of $J=2$.}\label{framework_multi}
\end{figure}

\subsection{FRAMEWORK OF MULTI-CHANNEL SAMPLING}

We now assume that graph signals are generated with $J$ generators, i.e.,
\begin{align}
\bm{x}=\sum_{\ell=0}^{J-1}\mb{A}_\ell \bm{d}_\ell,\label{eq:multi_model}
\end{align}
where $\mb{A}_\ell\in\mathbb{R}^{N\times K_\ell}$ and $\bm{d}_\ell\in\mathbb{R}^{K_\ell}$ are the $\ell$th generator and expansion coefficients, respectively. 

We often encounter the graph signal model in \eqref{eq:multi_model} for many applications. For example, in multiscale analysis of graph signals \cite{shuman_multiscale_2016}, piecewise-smooth graph signals are often considered. These signals are composed of a combination of piecewise constant and globally smooth signal models. This situation includes sensor network data \cite{chen_monitoring_2016}.

Suppose that the number of the sampling set in the $\ell$th channel is $K_\ell$.
The sampling operator $\mb{S}^\mathsf{T}\in\mathbb{R}^{(K_0+\cdots +K_{J-1})\times N}$ is given by
\begin{align}
\mb{S}^\mathsf{T}=
\begin{bmatrix}
 \mb{S}_0^\mathsf{T} \\
 \vdots\\
 \mb{S}_{J-1}^\mathsf{T}
\end{bmatrix}=
\begin{bmatrix}
\mb{I}_{\mathcal{M}_0\mathcal{V}} \mb{G}_0 \\
\vdots\\
\mb{I}_{\mathcal{M}_{J-1}\mathcal{V}} \mb{G}_1
\end{bmatrix},
\end{align}
where $\mathcal{M}_\ell$ is the $\ell$th sampling set and $\mb{G}_\ell\in\mathbb{R}^{N\times N}$ is an arbitrary graph filter for the $\ell$th subband.

According to \eqref{eq:recov_subspace}, the recovered graph signals are given by \cite[pp. 226--235]{eldar_sampling_2015}
\begin{align}
\tilde{\bm{x}}=
\begin{bmatrix}
\mb{A}_0 & \cdots &\mb{A}_{J-1}
\end{bmatrix}
\mb{M}_{S,A}^\dagger
\begin{bmatrix}
\mb{S}_0^\mathsf{T} \\ \vdots \\ \mb{S}_{J-1}^\mathsf{T}
\end{bmatrix}\bm{x},\label{eq:multi_samp}
\end{align}
where 
\begin{align}
\mb{M}_{S,A}=
\begin{bmatrix}
\mb{S}_0^\mathsf{T}\mb{A}_0 & \cdots & \mb{S}_0^\mathsf{T}\mb{A}_{J-1}\\
\vdots & \ddots & \vdots\\
\mb{S}_{J-1}^\mathsf{T}\mb{A}_0 & \cdots & \mb{S}_{J-1}^\mathsf{T}\mb{A}_{J-1}
\end{bmatrix}.\label{eq:corr_multi}
\end{align}
This implies recovery performance essentially depends on the invertibility of $\mb{M}_{S,A}$. 

Hereafter, we focus on the two-channel sampling, i.e., $J=2$, for simplicity.
To extend the following MCS for $J>2$, we may cascade the two-channel MCS \cite{sakiyama_two-channel_2019}, and perform the proposed sampling scheme recursively.

\subsection{SUBBAND-WISE REPRESENTATION OF MCS}

Since the MCS requires the matrix inversion of \eqref{eq:corr_multi}, it could be computational consuming especially for large graphs.
To reduce its computational cost, we rewrite \eqref{eq:corr_multi} as a computationally-efficient form.

For simplicity, this paper focuses on the critically-sampled case, i.e., $\mathcal{M}_1=\mathcal{M}_0^c$ (please refer to the notation in Section II-\ref{sec:sss_single}), but any sampling ratio can be applied to our MCS including over- and under-sampled cases.
Suppose that $\mb{S}_0^\mathsf{T}\mb{A}_0$ and $\mb{S}_1^\mathsf{T}\mb{A}_1$ are invertible. Under this condition, $\mb{M}_{S,A}^\dagger\mb{S}^\mathsf{T}$ can be rewritten as \cite{zhang_schur_2006}

\begin{align}
&\mb{M}_{S,A}^\dagger\mb{S}^\mathsf{T}\nonumber\\
&=
\begin{bmatrix}
\mb{S}_A^\mathsf{T}\mb{A}_0 & \mb{0}\\
\mb{0} & \mb{S}_B^\mathsf{T}\mb{A}_1
\end{bmatrix}^\dagger\nonumber\\
&\hspace{2.5ex}\times 
\begin{bmatrix}
\mb{I} & -\mb{S}_0^\mathsf{T}\mb{A}_1(\mb{S}_1^\mathsf{T}\mb{A}_1)^{-1}\\
-\mb{S}_1^\mathsf{T}\mb{A}_0(\mb{S}_0^\mathsf{T}\mb{A}_0)^{-1} & \mb{I}
\end{bmatrix}
\begin{bmatrix}
\mb{S}_0^\mathsf{T} \\ \mb{S}_1^\mathsf{T}
\end{bmatrix}\nonumber\\
&=\begin{bmatrix}
\mb{S}_A^\mathsf{T}\mb{A}_0 & \mb{0}\\
\mb{0} & \mb{S}_B^\mathsf{T}\mb{A}_1
\end{bmatrix}^\dagger
\begin{bmatrix}
\mb{S}_A^\mathsf{T} \\ \mb{S}_B^\mathsf{T}
\end{bmatrix}
,\label{eq:corr_inv_woodbury_2}
\end{align}
where
\begin{align}
\begin{split}
\mb{S}_A^\mathsf{T}\coloneqq&
\mb{S}_0^\mathsf{T}-\mb{S}_0^\mathsf{T}\mb{A}_1(\mb{S}_1^\mathsf{T}\mb{A}_1)^{-1}\mb{S}_1^\mathsf{T}\\
\mb{S}_B^\mathsf{T}\coloneqq&\mb{S}_1^\mathsf{T}-\mb{S}_1^\mathsf{T}\mb{A}_0(\mb{S}_0^\mathsf{T}\mb{A}_0)^{-1}\mb{S}_0^\mathsf{T}.
\end{split}\label{eq:samp_multi_redefine}
\end{align}
By viewing $\mb{S}_A^\mathsf{T}$ and $\mb{S}_B^\mathsf{T}$ as new sampling operators, \eqref{eq:multi_samp} can be expressed by
\begin{align}
\tilde{\bm{x}}=
\begin{bmatrix}
\mb{A}_0 & \mb{A}_1
\end{bmatrix}{\widetilde{\mb{M}}}_{S,A}^\dagger
\begin{bmatrix}
\mb{S}_A^\mathsf{T} \\ \mb{S}_B^\mathsf{T}
\end{bmatrix}\bm{x},\label{eq:multi_samp_sws}
\end{align}
where
\begin{align}
\widetilde{\mb{M}}_{S,A}=
\begin{bmatrix}
\mb{S}_A^\mathsf{T}\mb{A}_0 & \bm{0}\\
\bm{0} & \mb{S}_B^\mathsf{T}\mb{A}_1
\end{bmatrix}.\label{eq:corr_multi_sws}
\end{align}
\noindent In comparison with \eqref{eq:multi_samp}, we notice that \eqref{eq:multi_samp_sws} can be viewed as the subband-wise MCS. The modified framework is illustrated in Fig. \ref{framework_multi_subband}.
Obviously, the inverse of \eqref{eq:corr_multi_sws} requires the lower computational complexity than that of \eqref{eq:corr_multi}.

\begin{figure}[!t]
\centering
  \includegraphics[width=1.\columnwidth]{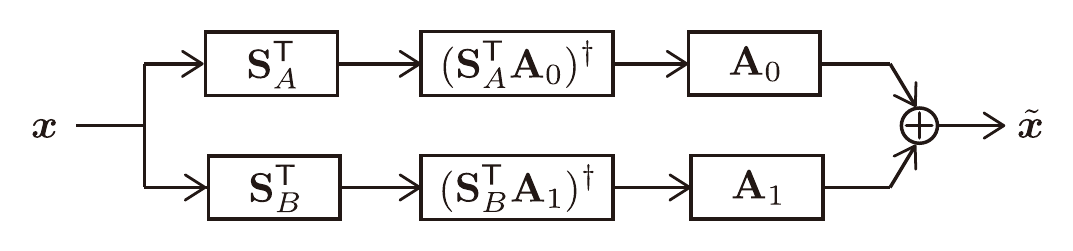}\vspace{-1em}
  \caption{The framework of the subband-wise MCS.}\label{framework_multi_subband}
\end{figure}

In contrast to the single channel setting, we need to simultaneously consider the best $\mathcal{M}_0$ and $\mathcal{M}_1$ in the multi-channel case.
In the following, we extend the SSS to that for MCS.

\subsection{SSS FOR MCS}

Based on \eqref{eq:multi_samp_sws}, we design the SSS for the graph MCS.
Recall that we assume that $\mb{S}^\mathsf{T}_0\mb{A}_0$ and $\mb{S}^\mathsf{T}_1\mb{A}_1$ in \eqref{eq:samp_multi_redefine} are invertible.
To satisfy this condition, we maximize the product of  $\det(\mb{S}^\mathsf{T}_0\mb{A}_0\mb{A}_0^\mathsf{T}\mb{S}_0)$ and $\det(\mb{S}^\mathsf{T}_1\mb{A}_1\mb{A}_1^\mathsf{T}\mb{S}_1)$ (see \eqref{eq:detmax1}).
We consider the following problem:
\begin{align}
\mathcal{M}^*=&\argmax_{\mathcal{M}\subset \mathcal{V}} \det([\mb{Z}_{0}]_\mathcal{M})\det([\mb{Z}_{1}]_{\mathcal{M}^c}),\label{eq:detmax1_multi}
\end{align}
where $\mb{Z}_0=\mb{G}_0\mb{A}_0\mb{A}_0^\mathsf{T}\mb{G}_0^\mathsf{T}$ and $\mb{Z}_1=\mb{G}_1\mb{A}_1\mb{A}_1^\mathsf{T}\mb{G}_1^\mathsf{T}$. By maximizing $\det([\mb{Z}_0]_\mathcal{M})\det([\mb{Z}_1]_{\mathcal{M}^c})$, the inverses of $\mb{S}^\mathsf{T}_0\mb{A}_0$ and $\mb{S}^\mathsf{T}_1\mb{A}_1$ become stable.
The cost function in \eqref{eq:detmax1_multi} is designed such that the graph signal is maximally separated into two subbands based on the expression in \eqref{eq:multi_samp_sws}. While the PR condition is not structurally guaranteed, it is known that PR can be practically realized in many cases \cite{chen_discrete_2015}.

By applying the Schur determinant formula \cite{zhang_schur_2006} to \eqref{eq:detmax1_multi}, the greedy SSS algorithm of the proposed graph MCS selects a node $y^*$ one-by-one that maximizes the following equation:
\begin{align}
y^*=&\argmax_{y\in\mathcal{M}^c}\det([\mb{Z}_0]_{\mathcal{M}\cup\{y\}})\det([\mb{Z}_1]_{\mathcal{M}\backslash\{y\}})\nonumber\\
=&\argmax_{y\in\mathcal{M}^c}\det([\mb{Z}_0]_{\mathcal{M}})\det([\mb{Z}_1]_{\mathcal{M}})\nonumber\\
&\cdot\left(\frac{[\mb{Z}_0]_{y,y}-[\mb{Z}_{0}]_{y,\mathcal{M}}([\mb{Z}_0]_{\mathcal{M},\mathcal{M}})^{-1}[\mb{Z}_{0}]_{\mathcal{M},y}}{[\mb{Z}_1]_{y,y}-[\mb{Z}_{1}]_{y,\overline{\mathcal{M}}}([\mb{Z}_1]_{\overline{\mathcal{M}},\overline{\mathcal{M}}})^{-1}[\mb{Z}_{1}]_{\overline{\mathcal{M}},y}}\right)\nonumber\\
=&\argmax_{y\in\mathcal{M}^c}\frac{[\mb{Z}_0]_{y,y}-[\mb{Z}_{0}]_{y,\mathcal{M}}([\mb{Z}_0]_{\mathcal{M},\mathcal{M}})^{-1}[\mb{Z}_{0}]_{\mathcal{M},y}}{[\mb{Z}_1]_{y,y}-[\mb{Z}_{1}]_{y,\overline{\mathcal{M}}}([\mb{Z}_1]_{\overline{\mathcal{M}},\overline{\mathcal{M}}})^{-1}[\mb{Z}_{1}]_{\overline{\mathcal{M}},y}},\label{eq:detmax2_multi}
\end{align}
where $\overline{\mathcal{M}}=\mathcal{V}\backslash(\mathcal{M}\cup \{y\})$. 
We omit the multiplication with $\det([\mb{Z}_0]_{\mathcal{M}})\det([\mb{Z}_1]_{\mathcal{M}})$ in the third equivalence because it does not depend on $y^*$.

We apply the Neumann series approximation to the invereses of $[\mb{Z}_0]_{\mathcal{M},\mathcal{M}}$ and $[\mb{Z}_1]_{\overline{\mathcal{M}},\overline{\mathcal{M}}}$ in \eqref{eq:detmax2_multi} (see Algorithm \ref{algo:prop_single}).
As a result, the proposed algorithm is described by Algorithm \ref{algo:prop_multi}.

\begin{algorithm}[!h]\label{algo:prop_multi}
\DontPrintSemicolon
\setlength{\abovedisplayskip}{0pt}
\setlength{\belowdisplayskip}{0pt}
  \KwInput{$\mb{Z}_0,\mb{Z}_1,\mathcal{M}=\emptyset, K, m=0$}
  Set $\alpha_0^*$ such that $\|\mb{I}-\alpha_0^*\mb{Z}_\mathcal{M}\|_2\leq 1$\\
  Set $\beta_0^*$ such that $\|\mb{I}-\beta_0^*\mb{Z}_{\overline{\mathcal{M}}}\|_2\leq 1$\\
  \While{$|\mathcal{M}|<K$}{
  Compute $\bm{\varepsilon}_y^{0}=\alpha^*_0[\mb{Z}_0]_\mathcal{M}\bm{\delta}_y$\\
  Compute $\bm{\vartheta}_y^{0}=\beta^*_0[\mb{Z}_1]_{\overline{\mathcal{M}}}\,\bm{\delta}_y$\\
    \While{$\|\bm{\varepsilon}_y^{0}-\alpha_m^*[\mb{Z}_0]_\mathcal{M}\bm{\varepsilon}_y^{m}\|\geq \gamma \land \|\bm{\vartheta}_y^{0}-\beta_m^*[\mb{Z}_1]_{\overline{\mathcal{M}}}\,\bm{\vartheta}_y^{m}\|\geq \gamma$ for some $\gamma>0$}{
      $\alpha_{m+1}^*\leftarrow \frac{(\bm{\varepsilon}^m_y)^\mathsf{T}[\mb{Z}_0]_\mathcal{M}\bm{\varepsilon}^m_y}{\|[\mb{Z}_0]_\mathcal{M}\bm{\varepsilon}^m_y\|^2}$\\
      $\beta_{m+1}^*\leftarrow \frac{(\bm{\vartheta}^m_y)^\mathsf{T}[\mb{Z}_1]_{\overline{\mathcal{M}}}\bm{\vartheta}^m_y}{\|[\mb{Z}_1]_{\overline{\mathcal{M}}}\bm{\vartheta}^m_y\|^2}$\\
      $\bm{\varepsilon}^{m+1}_y \leftarrow \bm{\varepsilon}_y^{0}+(\mb{I}-\alpha^*_{m+1}[\mb{Z}_0]_\mathcal{M})\bm{\varepsilon}_y^{m}$\\
      $\bm{\vartheta}^{m+1}_y \leftarrow \bm{\vartheta}_y^{0}+(\mb{I}-\beta^*_{m+1}[\mb{Z}_1]_{\overline{\mathcal{M}}})\,\bm{\vartheta}_y^{m}$\\
      $m\leftarrow m+1$\\
     }
    $y^*\rightarrow\mathop{\arg\max}_{y\in\mathcal{M}^c} \frac{[\mb{Z}_0]_{yy}-[\mb{Z}_0]_{y,\mathcal{M}}\bm{\varepsilon}_y^{m}}{[\mb{Z}_1]_{yy}-[\mb{Z}_1]_{y,\overline{\mathcal{M}}}\bm{\vartheta}_y^{m}}$\\
    $\mathcal{M}\leftarrow \mathcal{M}\cup \{y\}$
  }
  \KwOutput{$\mathcal{M}$}
\caption{Two-channel SSS}
\end{algorithm}

In the following, we clarify relationship between the above MCS and existing BFBs.

\section{RELATIONSHIP BETWEEN MCS AND BGFB}

In this section, we reveal that existing BGFBs are a special case of our MCS in \eqref{eq:multi_samp}. We depict the framework of PR BGFBs in Fig. \ref{framework_multi_gfb}.

Let $\mathcal{G}_\text{bpt}=(\mathcal{V}_L,\mathcal{V}_H,\mathcal{E})$ be a bipartite graph, where $\mathcal{V}_L$ and $\mathcal{V}_H$ are two disjoint sets of vertices such that every edge is connected between a vertex in $\mathcal{V}_L$ and that in $\mathcal{V}_H$.
In other words, no edges exist within $\mathcal{V}_L$ and $\mathcal{V}_H$.
We denote the normalized Laplacian matrix for $\mathcal{G}_\text{bpt}$ by $\bm{\mathcal{L}}_{\text{bpt}}$. Its eigendecomposition is given by $\bm{\mathcal{L}}_\text{bpt}=\mb{V}\mb{\Lambda}_\text{bpt}\mb{V}^\mathsf{T}$, where
\begin{align}
\mb{V}\coloneqq
\begin{bmatrix}
\mb{U}_{LL} & \mb{U}_{LL}\\
\mb{U}_{HL} & -\mb{U}_{HL}
\end{bmatrix}.
\end{align}
Let $\mb{S}^\mathsf{T}_{\text{ana},\ell}$ and $\mb{S}_{\text{syn},\ell}$ be the $\ell$th analysis and synthesis transforms, respectively. Following Theorem 2 in \cite{sakiyama_two-channel_2019}, analysis and synthesis transforms in a BGFB can be expressed by 
\begin{align}
\begin{split}
\mb{S}_{\text{ana},\ell}^\mathsf{T}\coloneqq&\mb{I}_{\mathcal{MV}}\mb{H}_\ell=
\mb{U}_{LL}
\begin{bmatrix}
\mb{I}_{N/2} & \mb{J}_{N/2}
\end{bmatrix}
\widehat{\mb{H}}_\ell
\mb{V}^\mathsf{T}\\
\mb{S}_{\text{syn},\ell}\coloneqq&\mb{G}_\ell\mb{I}_{\mathcal{M}^c\mathcal{V}}^\mathsf{T}=
\mb{V}
\widehat{\mb{G}}_\ell
\begin{bmatrix}
\mb{I}_{N/2} \\ -\mb{J}_{N/2}
\end{bmatrix}
\mb{U}_{LL}^\mathsf{T}
\end{split}\, \text{for }\ell=1, 2,\label{eq:bpt_samp}
\end{align}
where $\widehat{\mb{H}}_\ell=H_\ell(\mb{\Lambda}_\text{bpt})$, $\widehat{\mb{G}}_\ell=G_\ell(\mb{\Lambda}_\text{bpt})$, and $\mb{J}$ is the counter-identity matrix.
Note that the PR condition in GFBs can be expressed by \cite{narang_compact_2013,narang_perfect_2012}
\begin{align}
\mb{S}_{\text{syn},0}\mb{S}_{\text{ana},0}^\mathsf{T}+\mb{S}_{\text{syn},1}\mb{S}_{\text{ana},1}^\mathsf{T}=\mb{I}.\label{eq:pr_cond_trad}
\end{align}
In graph filter bank designs, $\widehat{\mb{H}}_\ell$ and $\widehat{\mb{G}}_\ell$ in \eqref{eq:bpt_samp} are designed so that \eqref{eq:pr_cond_trad} is satisfied. 

From a perspective of generalized sampling \cite{eldar_sampling_2015} (cf. Fig.~\ref{framework_multi_subband}), we can view that  $\mb{S}_{\text{ana},0}^\mathsf{T}=\mb{S}_A^\mathsf{T}$, $\mb{S}_{\text{ana},1}^\mathsf{T}=\mb{S}_B^\mathsf{T}$, $\mb{S}_{\text{syn},0}=\mb{A}_0(\mb{S}_A^\mathsf{T}\mb{A}_0)^{-1}$, and $\mb{S}_{\text{syn},1}=\mb{A}_1(\mb{S}_B^\mathsf{T}\mb{A}_1)^{-1}$.
Consequently, the PR condition of GFBs in \eqref{eq:pr_cond_trad} is rewritten as a MCS as follows:
\begin{align}
\widetilde{\mb{A}}_0\mb{S}_A^\mathsf{T}+\widetilde{\mb{A}}_1\mb{S}_B^\mathsf{T}=\mb{I},\label{eq:pr_cond_mcs}
\end{align}
where $\widetilde{\mb{A}}_0=\mb{A}_0(\mb{S}^\mathsf{T}_A\mb{A}_0)^{-1}$ and $\widetilde{\mb{A}}_1=\mb{A}_1(\mb{S}^\mathsf{T}_B\mb{A}_1)^{-1}$.
While the GFB form \eqref{eq:pr_cond_trad} and the MCS form \eqref{eq:pr_cond_mcs} are not identical in general, they coincide each other with bipartite graphs.
This can be stated in the following theorem.

\begin{theorem}\label{theorem:multi_pr_cond}
Let $\mb{S}_0^\mathsf{T}\in\mathbb{R}^{K\times N}$ and $\mb{S}^\mathsf{T}_1\in \mathbb{R}^{(N-K)\times}N$ be two sampling operators, and let $\mb{A}_0\in\mathbb{R}^{N\times K}$ and $\mb{A}_1\in\mathbb{R}^{K\times(N-K)}$ be generation operators such that $\mb{S}^\mathsf{T}_0\mb{A}_0$ and $\mb{S}^\mathsf{T}_1\mb{A}_1$ are invertible. If sampling and reconstruction are performed on the bipartite graph defined by \eqref{eq:bpt_samp}, it follows that
\begin{align}
\widetilde{\mb{A}}_0\mb{S}_A^\mathsf{T}+\widetilde{\mb{A}}_1\mb{S}_B^\mathsf{T}=\mb{S}_{\text{\emph{syn}},0}\mb{S}_{\text{\emph{ana}},0}^\mathsf{T}+\mb{S}_{\text{\emph{syn}},1}\mb{S}_{\text{\emph{ana}},1}^\mathsf{T}=\mb{I}.\label{eq:pr_cond_bpt}
\end{align}
That is, the PR condition for the graph MCS is reduced to that of the conventional BGFBs as long as the graph is bipartite and the normalized Laplacian is used as the graph operator.
\end{theorem}

\begin{proof}
By definition in \eqref{eq:samp_multi_redefine}, we need to show $\widetilde{\mb{A}}_0\mb{S}^\mathsf{T}_0\widetilde{\mb{A}}_1\mb{S}^\mathsf{T}_1+\widetilde{\mb{A}}_1\mb{S}^\mathsf{T}_1\widetilde{\mb{A}}_0\mb{S}^\mathsf{T}_0=\mb{0}$. This can be verified by
\begin{align}
&\mb{V}^\mathsf{T}
(\widetilde{\mb{A}}_0\mb{S}_0^\mathsf{T}\widetilde{\mb{A}}_1\mb{S}_1^\mathsf{T}+\widetilde{\mb{A}}_1\mb{S}_1^\mathsf{T}\widetilde{\mb{A}}_0\mb{S}_0^\mathsf{T})
\mb{V}
\nonumber\\
&=2(\widehat{\mb{G}}_0\begin{bmatrix}\mb{I} & \mb{J}\end{bmatrix}\widehat{\mb{H}}_0\widehat{\mb{G}}_1\begin{bmatrix}\mb{I} &  -\mb{J}\end{bmatrix}\widehat{\mb{H}}_1)\nonumber\\
&=2(\widehat{\mb{G}}_0\widehat{\mb{H}}_0\widehat{\mb{G}}_1\widehat{\mb{H}}_1+\widehat{\mb{G}}_0\widehat{\mb{H}}_0^\prime\widehat{\mb{G}}_1\widehat{\mb{H}}_1\nonumber\\
&\hspace{2.5ex}-\widehat{\mb{G}}_0\widehat{\mb{H}}_0\widehat{\mb{G}}_1\widehat{\mb{H}}_1^\prime-\widehat{\mb{G}}_0\widehat{\mb{H}}_0\widehat{\mb{G}}_1\widehat{\mb{H}}_1)\nonumber\\
&=2(\widehat{\mb{G}}_0\widehat{\mb{H}}_0^\prime\widehat{\mb{G}}_1\widehat{\mb{H}}_1-\widehat{\mb{G}}_0\widehat{\mb{H}}_0\widehat{\mb{G}}_1\widehat{\mb{H}}_1^\prime)\nonumber\\
&=2\widehat{\mb{G}}_0\widehat{\mb{G}}_1(\widehat{\mb{H}}_0^\prime\widehat{\mb{H}}_1-\widehat{\mb{H}}_0\widehat{\mb{H}}_1^\prime)\nonumber\\
&=2\widehat{\mb{G}}_0\widehat{\mb{G}}_1(\widehat{\mb{H}}_0\mb{J}\widehat{\mb{H}}_1-\widehat{\mb{H}}_0\mb{J}\widehat{\mb{H}}_1)=\mb{0},
\end{align}
where $\widehat{\mb{H}}^\prime=\mb{J}\widehat{\mb{H}}=\widehat{\mb{H}}\mb{J}$. This completes the proof.
\end{proof}

\begin{figure}[!t]
\centering
  \includegraphics[width=1.\columnwidth]{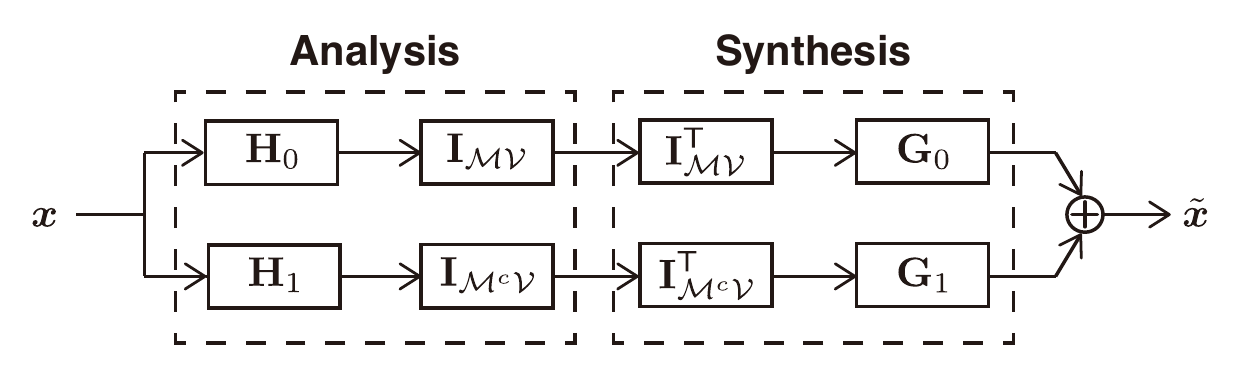}\vspace{-1em}
  \caption{Framework of two-channel PR GFBs.}\label{framework_multi_gfb}
\end{figure}

\noindent Note that \eqref{eq:pr_cond_bpt} can be expressed as
\begin{align}
&\begin{bmatrix}
\mb{A}_0 & \mb{A}_1
\end{bmatrix}
\begin{bmatrix}
\mb{S}_A^\mathsf{T}\mb{A}_0 & \bm{0}\\
\bm{0} & \mb{S}_B^\mathsf{T}\mb{A}_1
\end{bmatrix}^{-1}
\begin{bmatrix}
\mb{S}_A^\mathsf{T} \\ \mb{S}_B^\mathsf{T}
\end{bmatrix}\nonumber\\
&=
\begin{bmatrix}
\mb{A}_0 & \mb{A}_1
\end{bmatrix}
\begin{bmatrix}
\mb{S}_0^\mathsf{T}\mb{A}_0 & \bm{0}\\
\bm{0} & \mb{S}_1^\mathsf{T}\mb{A}_1
\end{bmatrix}^{-1}
\begin{bmatrix}
\mb{S}_0^\mathsf{T} \\ \mb{S}_1^\mathsf{T}
\end{bmatrix}\nonumber\\
&=
\begin{bmatrix}
\mb{S}_{\text{syn},0} & \mb{S}_{\text{syn},1}
\end{bmatrix}
\begin{bmatrix}
\mb{S}_{\text{ana},0}^\mathsf{T} \\ \mb{S}_{\text{ana},1}^\mathsf{T}
\end{bmatrix}\nonumber\\
&=\mb{I}.\label{eq:pr_cond_sbs}
\end{align}
For the second equality, we replace the notations with $\mb{S}_{\text{ana},0}^\mathsf{T}=\mb{S}_0^\mathsf{T}$, $\mb{S}_{\text{ana},1}^\mathsf{T}=\mb{S}_1^\mathsf{T}$, $\mb{S}_{\text{syn},0}=\mb{A}_0(\mb{S}_0^\mathsf{T}\mb{A})^{-1}$, and  $\mb{S}_{\text{ana},1}=\mb{A}_1(\mb{S}_1^\mathsf{T}\mb{A}_1)^{-1}$ \cite{christensen_oblique_2004}. 
According to Theorem 1 and \eqref{eq:pr_cond_sbs}, we directly obtain  $\mb{S}_A^\mathsf{T}=\mb{S}_0^\mathsf{T}$ and $\mb{S}_B^\mathsf{T}=\mb{S}_1^\mathsf{T}$ if sampling and reconstruction are performed on the bipartite graph.
This statement also follows for graph spectral sampling \cite{sakiyama_two-channel_2019}.

\section{RECOVERY EXPERIMENTS}

\begin{figure}[t!]
\centering
 \subfigure[][Original]
  {\centering\includegraphics[width=0.47\linewidth]{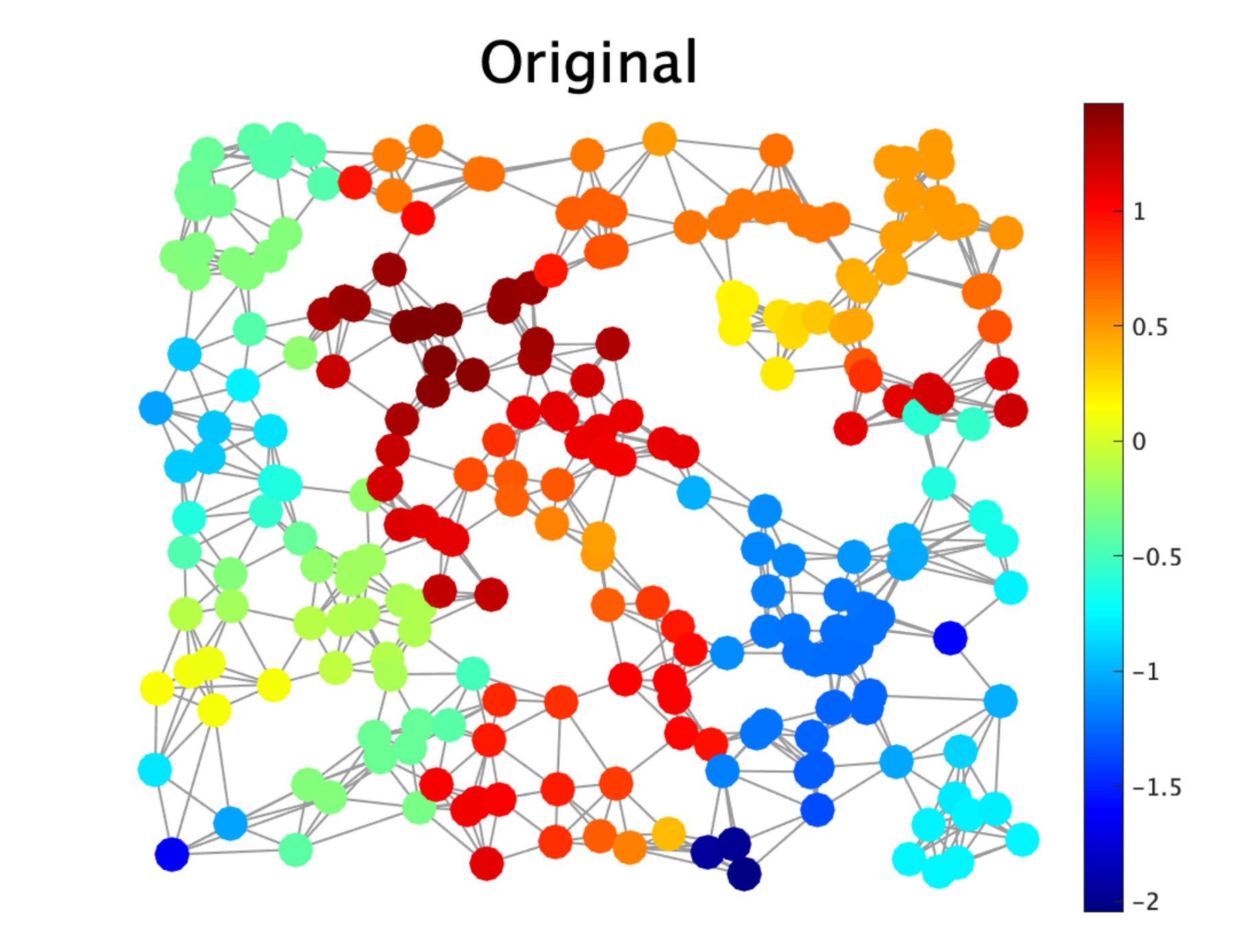} \label{mbp_ori}}
 \subfigure[][Proposed]
  {\centering\includegraphics[width=0.47\linewidth]{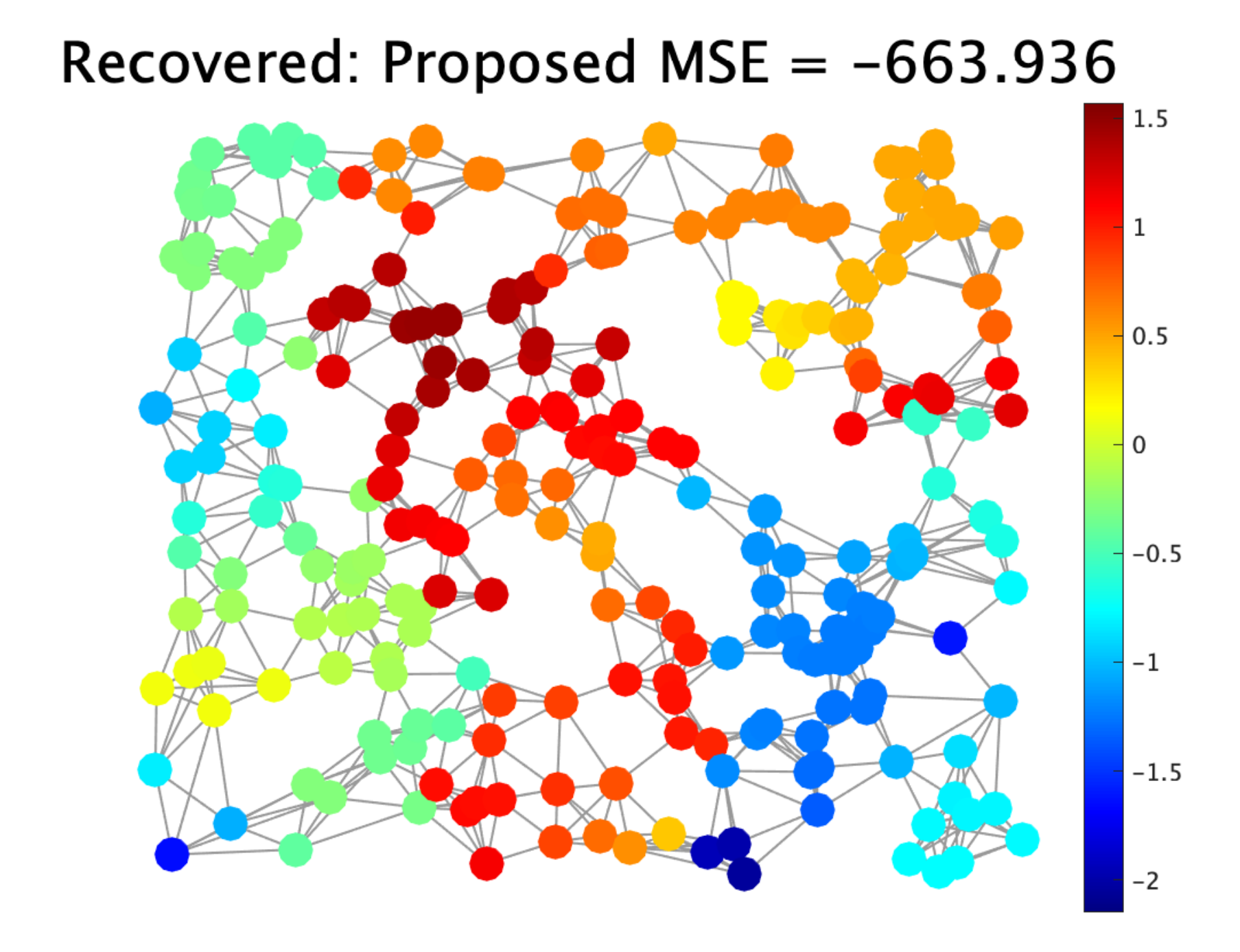} \label{pws_prop}}
 \subfigure[][Recon. w/ ch. 1]
  {\centering\includegraphics[width=0.47\linewidth]{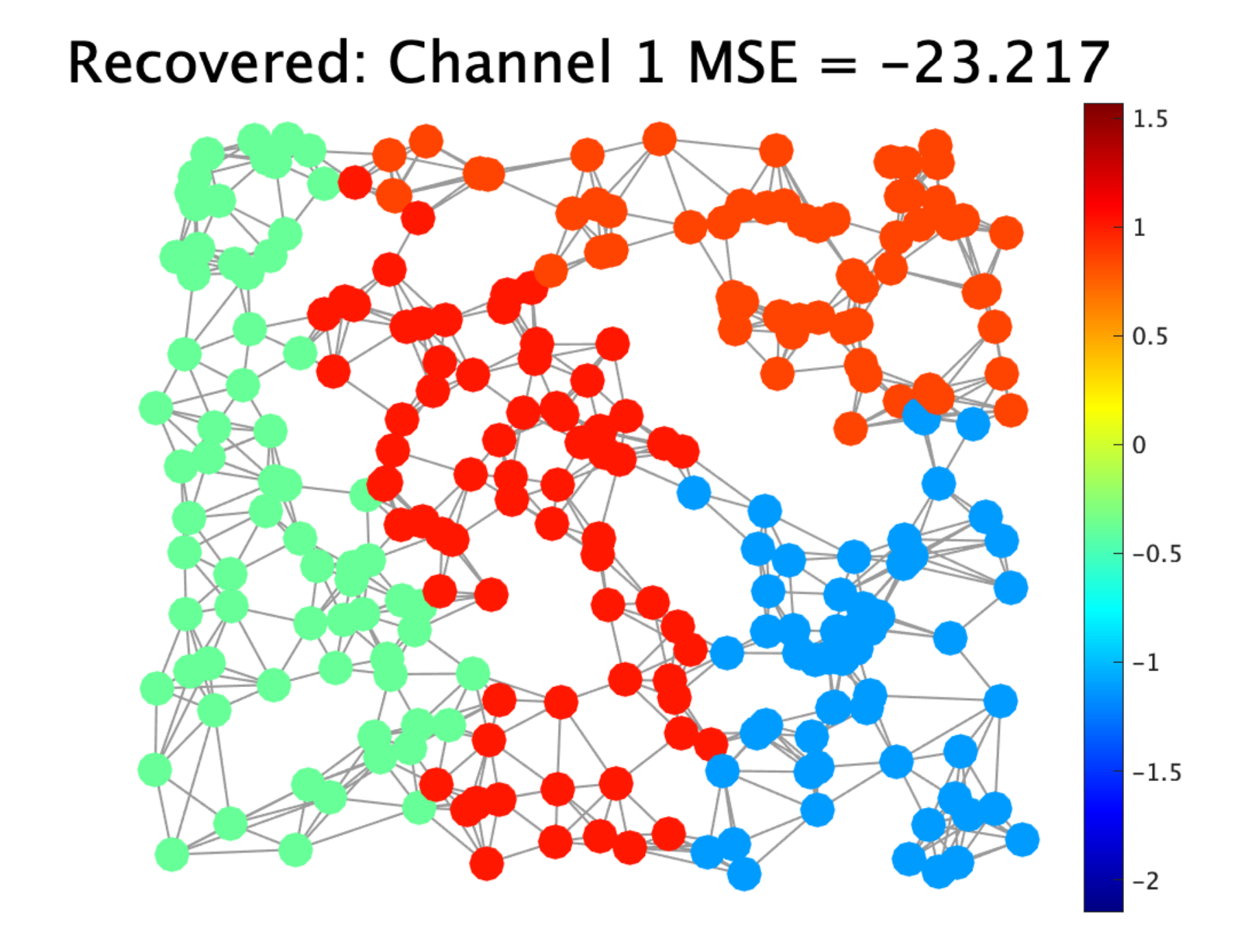} \label{pws_ch1}}
 \subfigure[][Recon. w/ ch. 2]
  {\centering\includegraphics[width=0.47\linewidth]{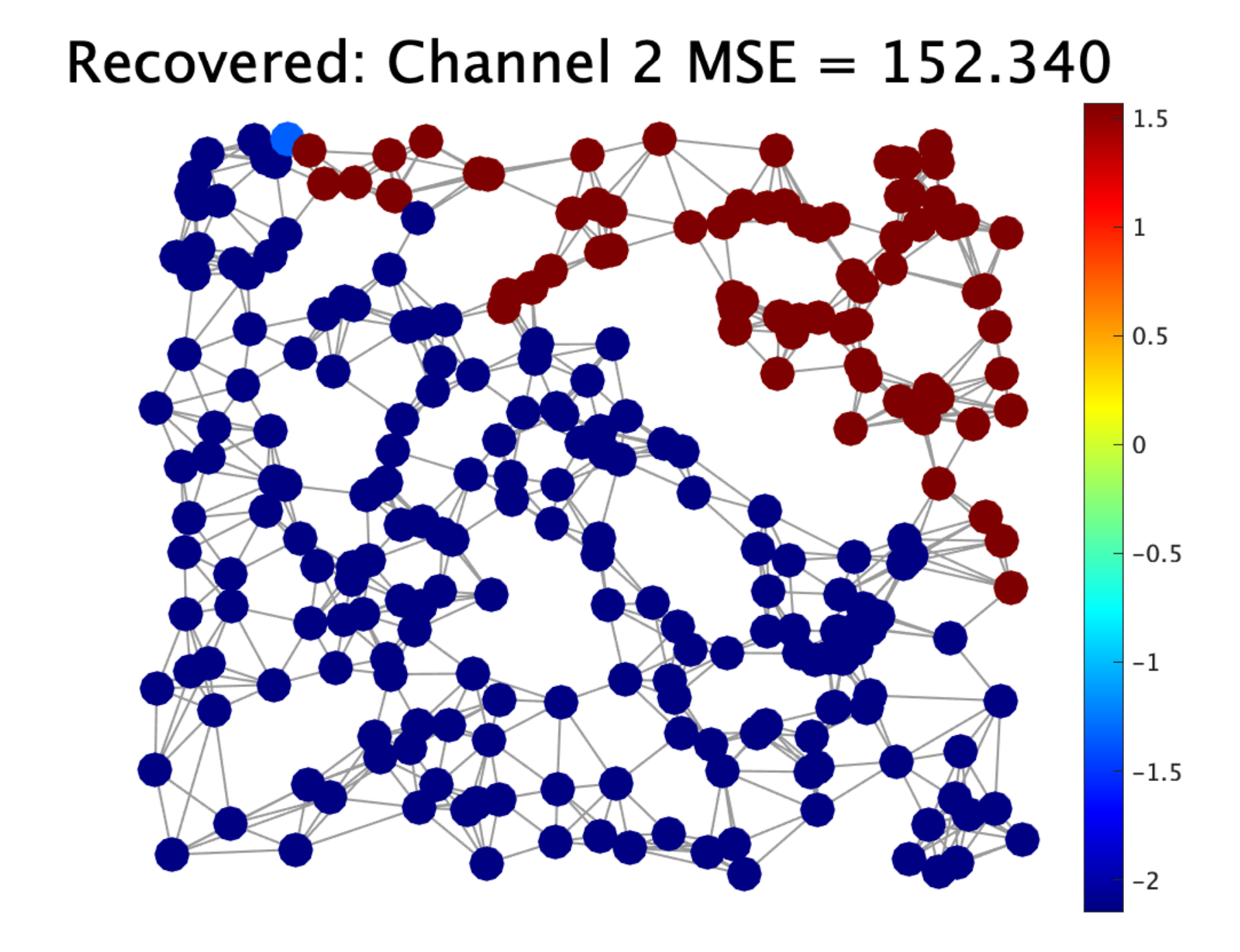} \label{pws_ch2}}
 \subfigure[][GraphQMF]
  {\centering\includegraphics[width=0.47\linewidth]{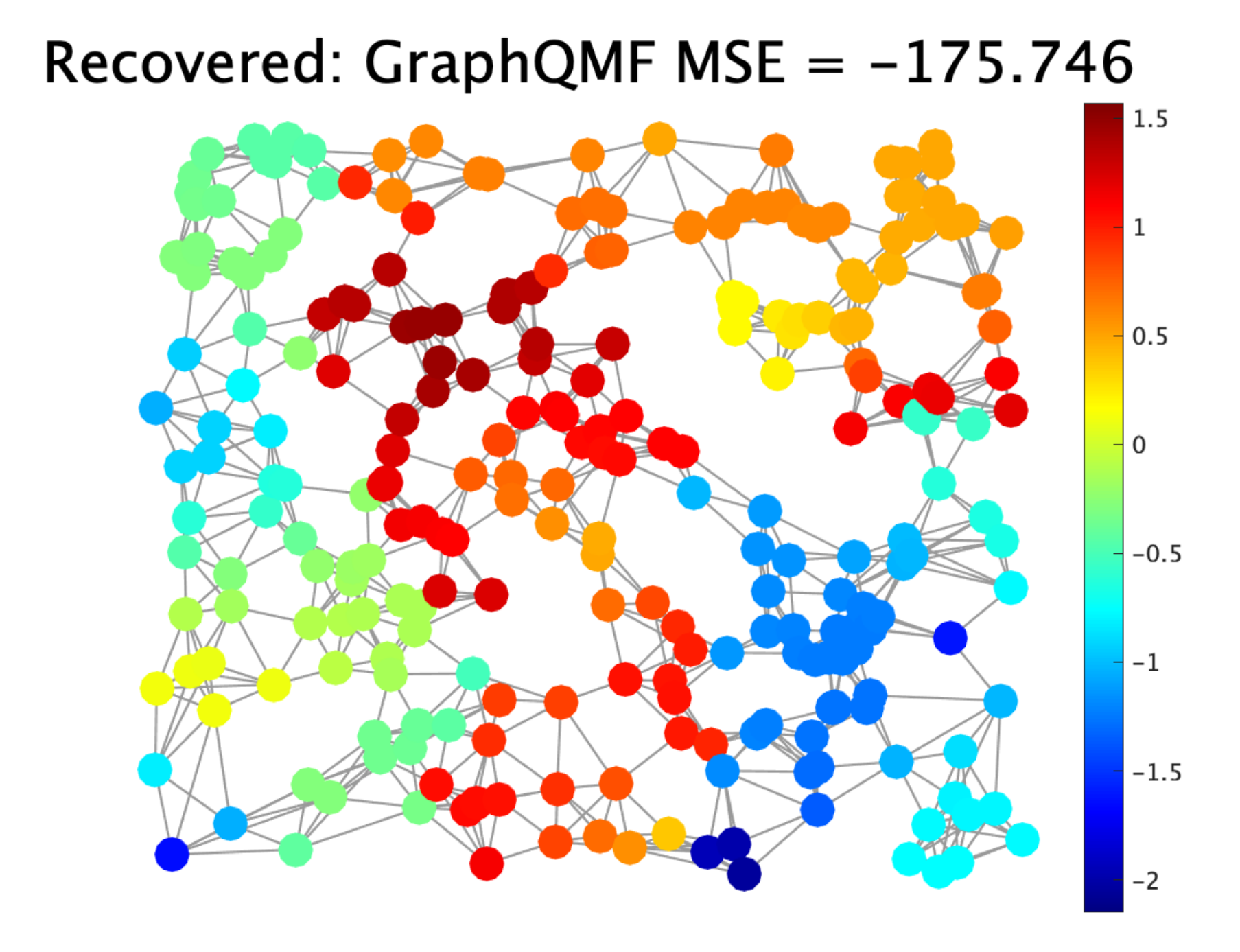} \label{pws_graphQMF}}
 \subfigure[][GraphBior]
  {\centering\includegraphics[width=0.47\linewidth]{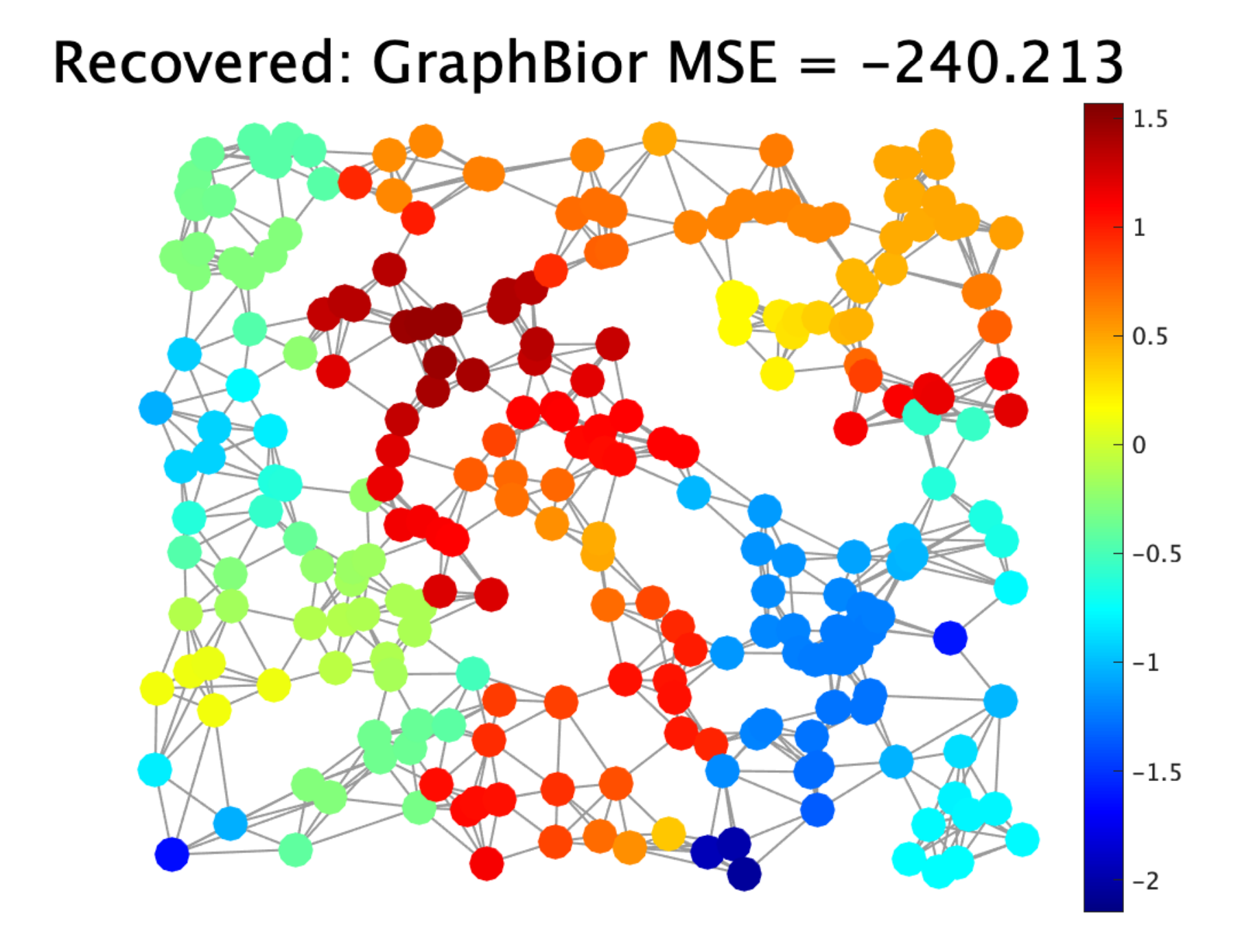} \label{pws_graphBior}}
\caption{Examples of recovery for PWS graph signals on sensor graphs.}
\label{exp:multi_recov_pws}
\end{figure}

\begin{figure}[t!]
\centering
 \subfigure[][Original]
  {\centering\includegraphics[width=0.47\linewidth]{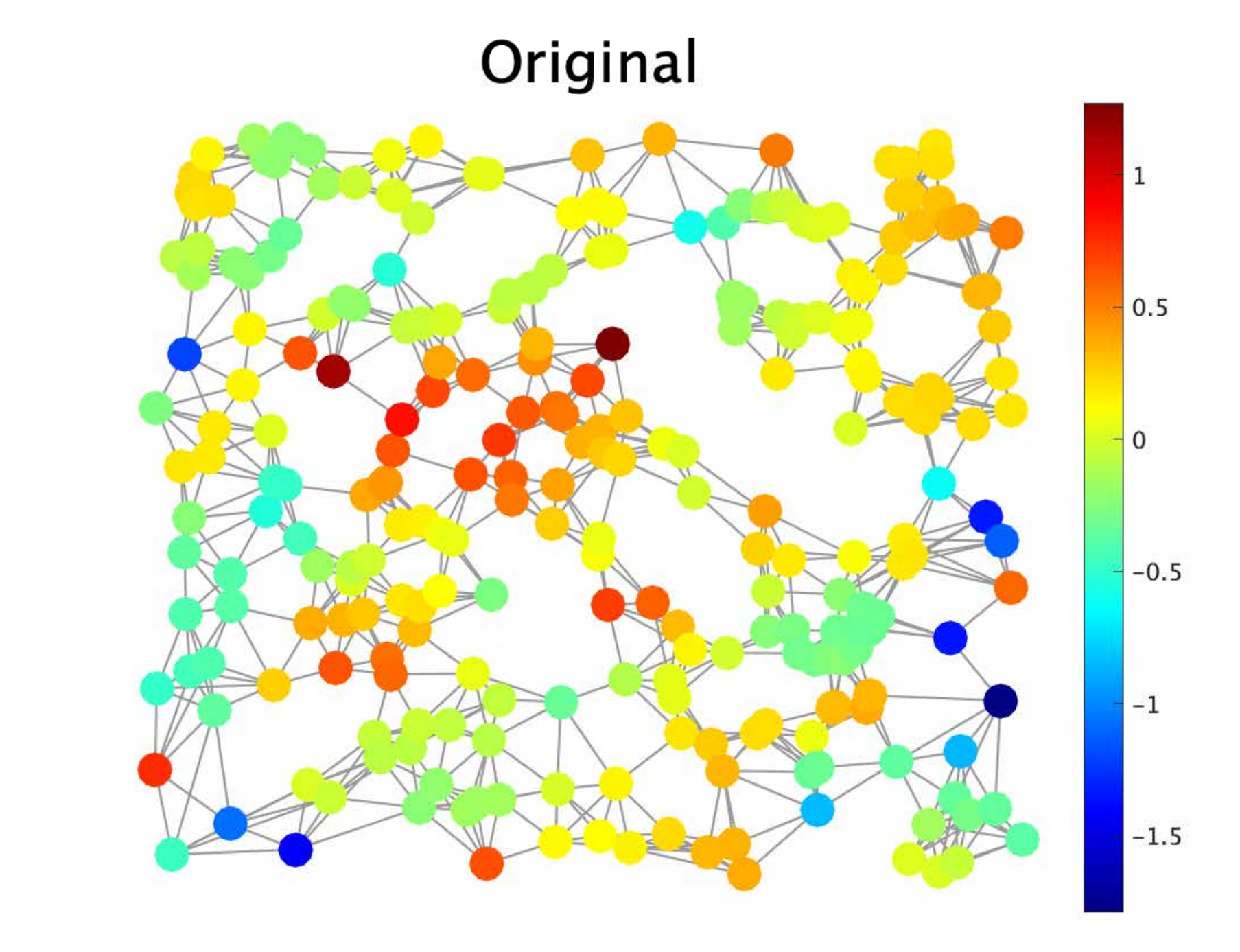} \label{ubp_ori}}
 \subfigure[][Proposed]
  {\centering\includegraphics[width=0.47\linewidth]{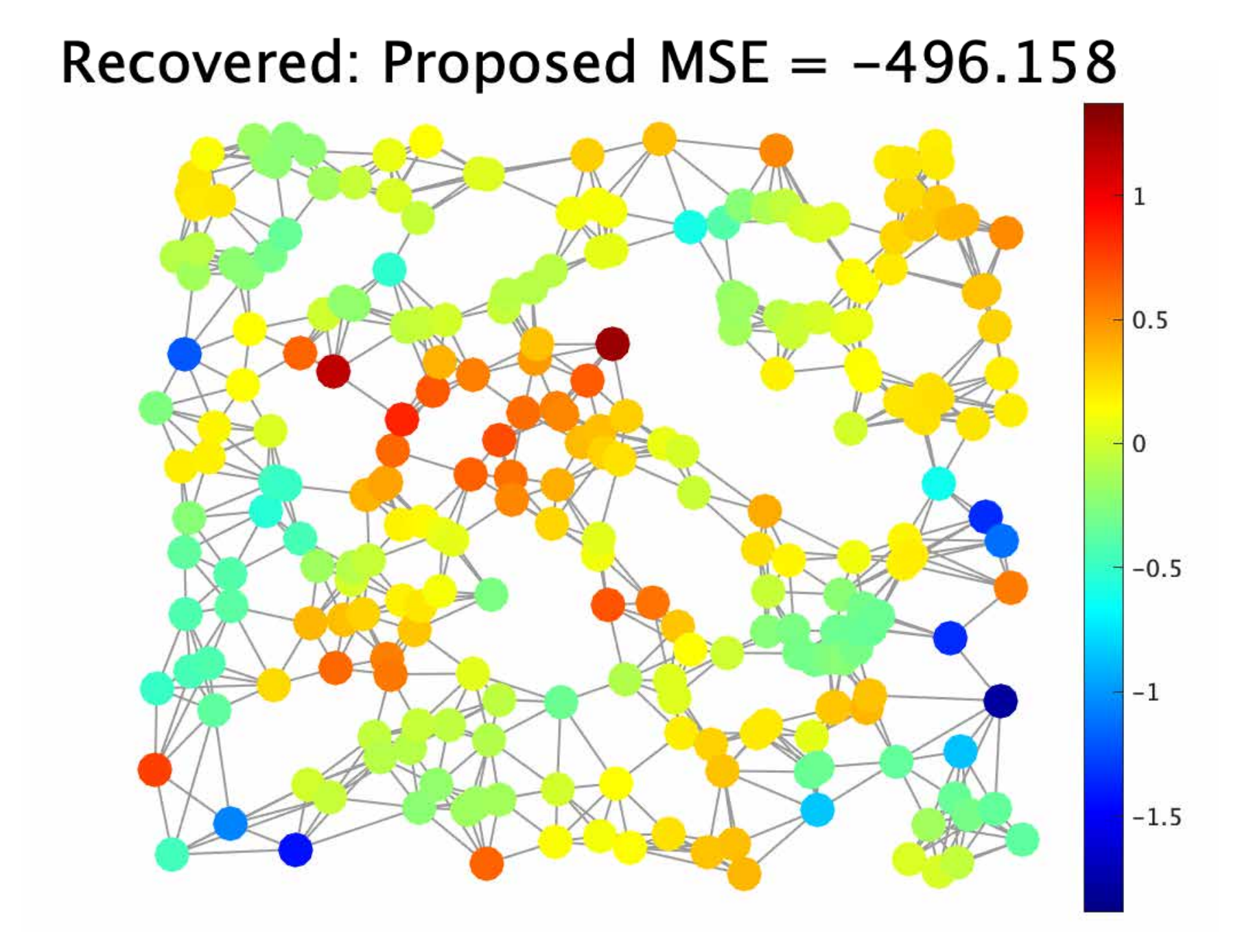} \label{ubp_prop}}
 \subfigure[][Recon. w/ ch. 1]
  {\centering\includegraphics[width=0.47\linewidth]{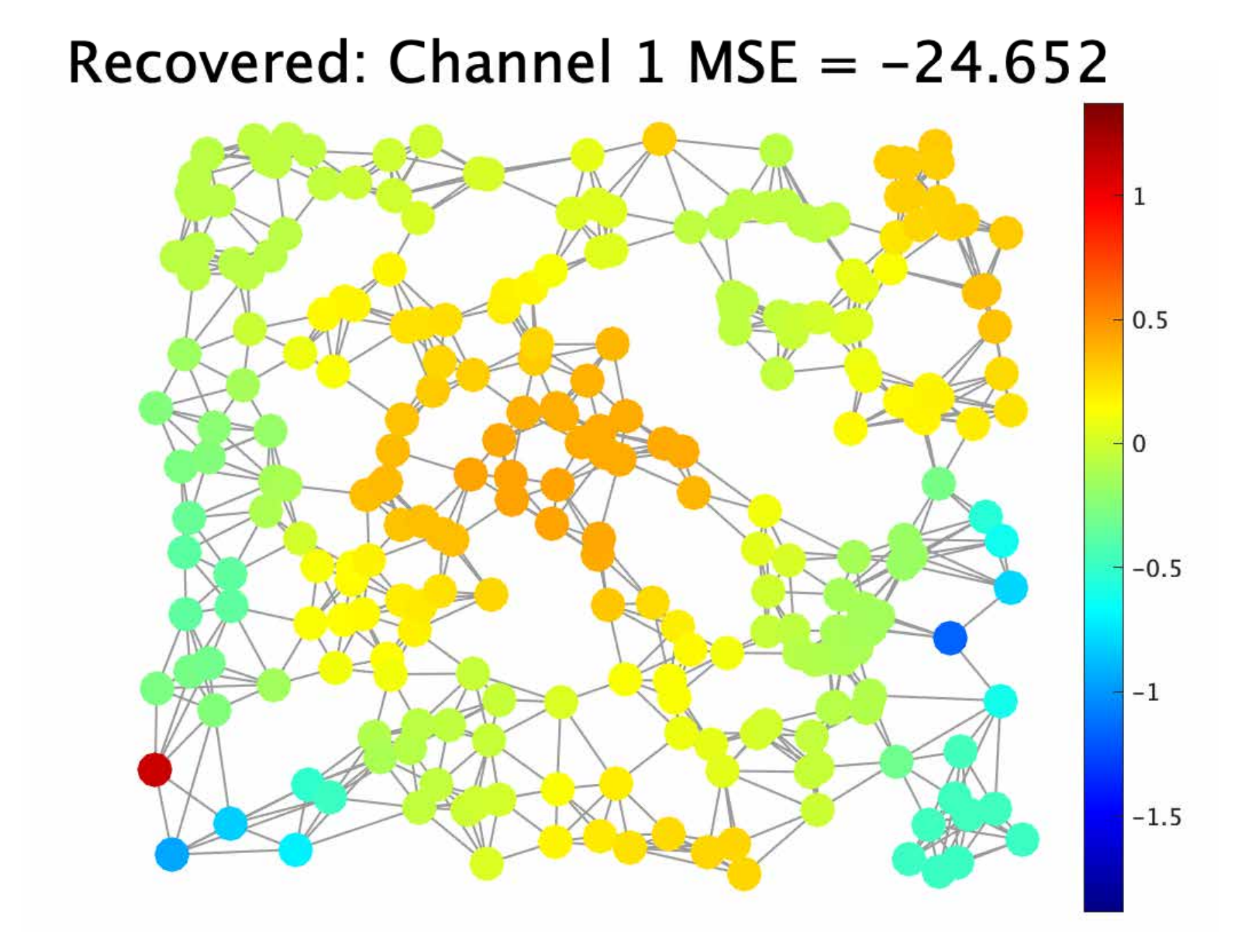} \label{ubp_ch1}}
 \subfigure[][Recon. w/ ch. 2]
  {\centering\includegraphics[width=0.47\linewidth]{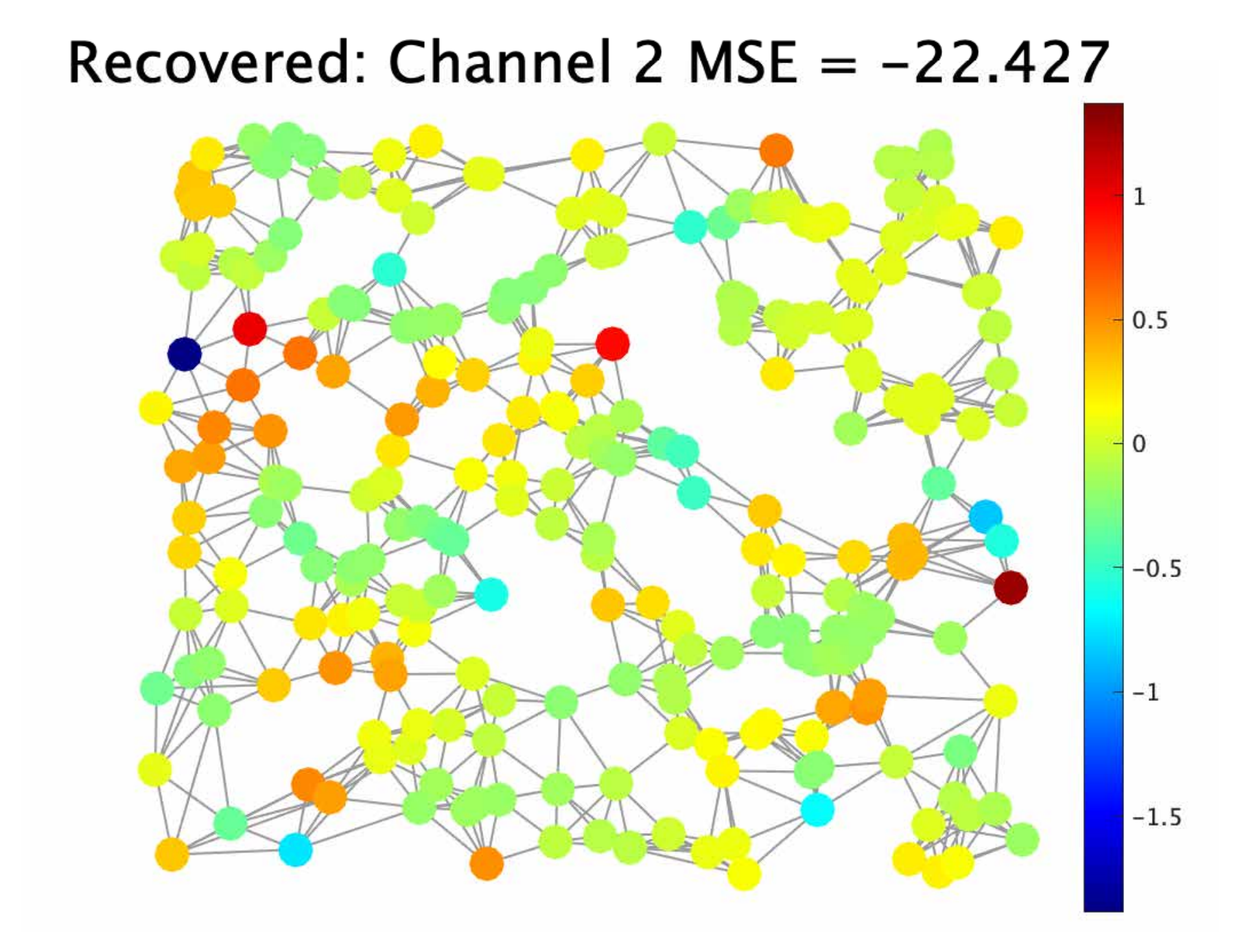} \label{ubp_ch2}}
 \subfigure[][GraphQMF]
  {\centering\includegraphics[width=0.47\linewidth]{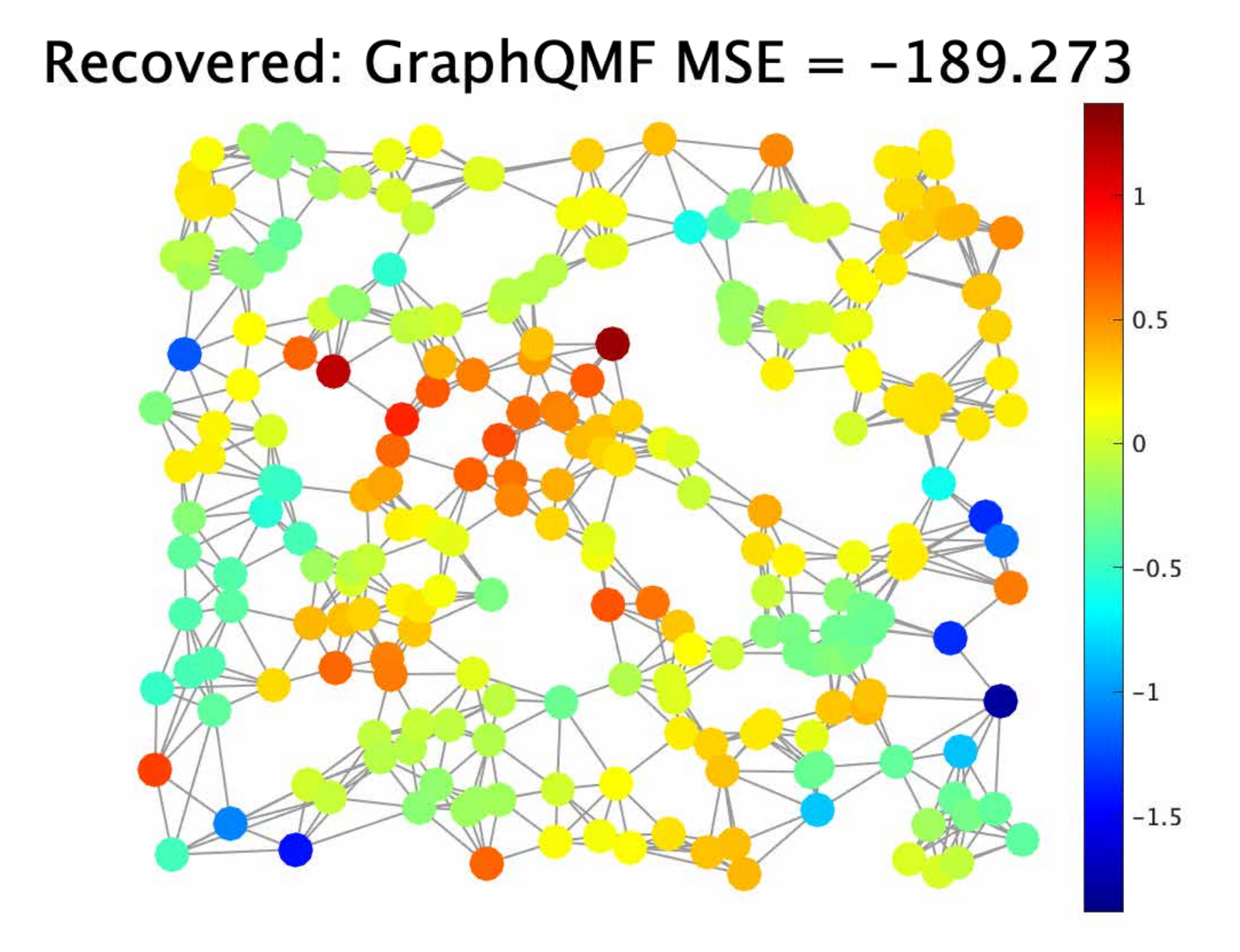} \label{ubp_graphQMF}}
 \subfigure[][GraphBior]
  {\centering\includegraphics[width=0.47\linewidth]{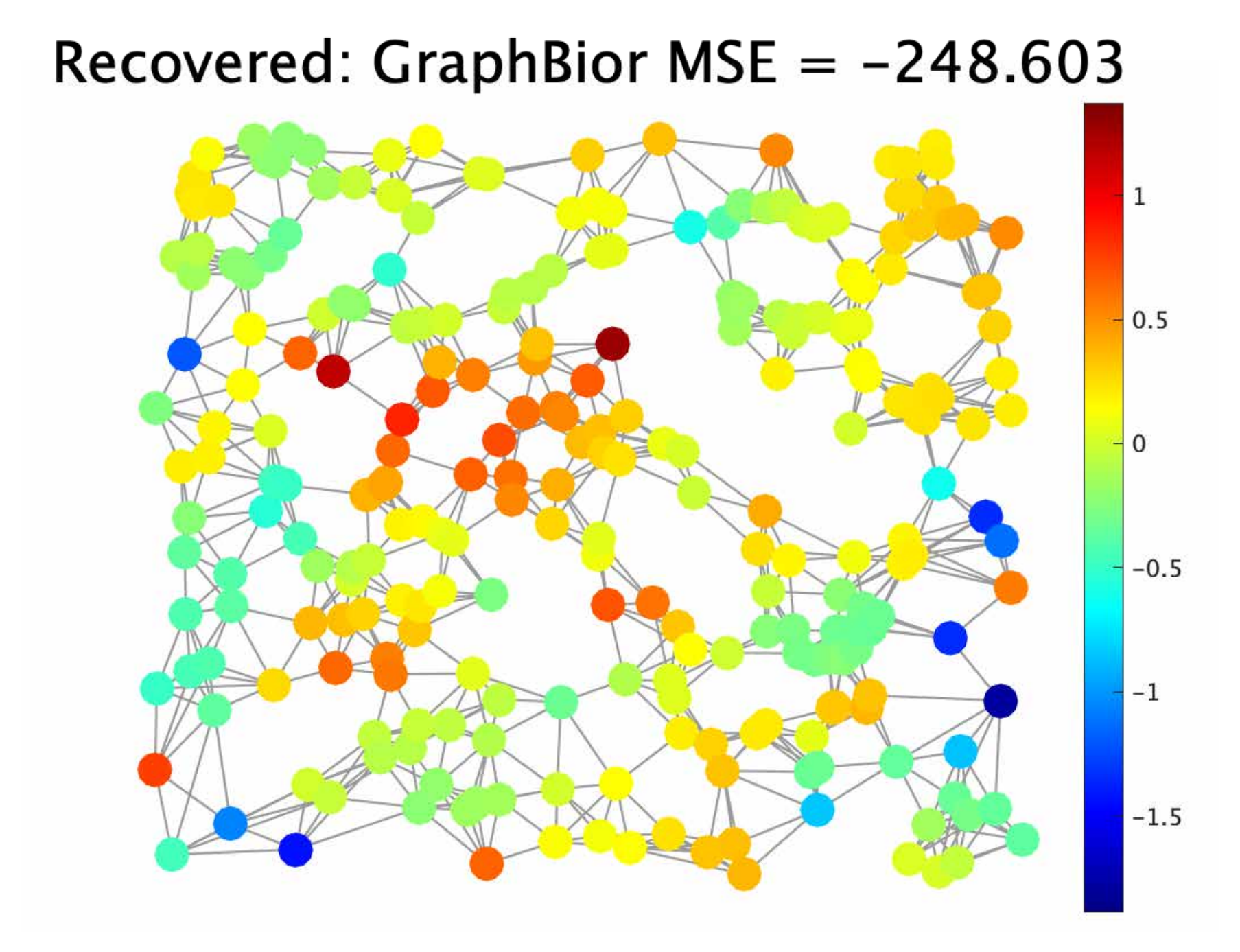} \label{ubp_graphBior}}
\caption{Examples of recovery for UBP graph signals on sensor graphs.}
\label{exp:multi_recov_ubp}
\end{figure}

In this section, we validate performances of the proposed SSS for the graph MCS. We perform graph signal recovery experiments with synthetic and real-world graphs.

\begin{table}[t!]
\centering
\caption{Average MSEs of 30 independent runs. RS is the random sensor graph and SR is the Swiss roll graph.}\label{tab:ave_mse}
\setlength{\tabcolsep}{0.12em}
\begin{tabular}{cc|ccccc}
\hline
\multicolumn{2}{c|}{Methods} & Proposed & \begin{tabular}[c]{@{}c@{}}Recon. \\ w/ ch.~1\end{tabular} & \begin{tabular}[c]{@{}c@{}}Recon. \\ w/ ch.~2\end{tabular} & GraphQMF & GraphBior \\ \hline\hline
\multicolumn{1}{c|}{\multirow{2}{*}{RS}} & PWS & \textbf{-619.04} & -18.81 & 194.68 & -168.25 & -230.60 \\
\multicolumn{1}{c|}{} & UBP & \textbf{-465.98} & -18.89 & 24.04 & -186.39 & -249.12 \\ \hline
\multicolumn{1}{c|}{\multirow{2}{*}{SR}} & PWS & \textbf{-654.56} & -19.93 & 174.18 & -172.95 & -231.03 \\
\multicolumn{1}{c|}{} & UBP & \textbf{-382.10} & -17.00 & -9.40 & -189.39 & -248.31 \\ \hline
\end{tabular}
\end{table}

\subsection{SYNTHETIC GRAPHS}

\subsubsection{SETUP}
The experiments are performed on random sensor and Swiss roll graphs with $N=256$.
The sampling ratio is set to $K=|\mathcal{M}_0|=N/2$. For both graphs, we generate two synthetic graph signals:
\begin{description}
\item[Piecewise smooth (PWS) graph signals\cite{chen_multiresolution_2018}:] It is composed of piecewise constant components and smooth components:
\begin{align}
\bm{x}= [\bm{1}_{\mathcal{T}_1}\cdots\bm{1}_{\mathcal{T}_P}]\bm{d}_1+\mb{U}_{\mathcal{B}}\bm{d}_2,\label{eq:pws_model}
\end{align}
where the number of clusters is set to $P=4$ and the bandwidth is set to $|\mathcal{B}|=K/4$.
\item[Union of band-pass (UBP) graph signals:] It is composed of several band-pass components:
\begin{align}
\bm{x}=\mb{U}\sum_{\ell=1}^2 G_\ell(\mb{\Lambda})\bm{d}_\ell,
\end{align}
where generation (synthesis) filters $G_\ell(\mb{\Lambda})$ are implemented by Meyer wavelet kernel\cite{perraudin_gspbox:_2014}. 
\end{description}
In both cases, analysis filters are given by Mexican hat wavelet kernel\cite{perraudin_gspbox:_2014}. 

We calculate the average MSE of reconstructed graph signals for 30 independent runs. We compare the result to the well-known BGFBs, graphQMF \cite{narang_perfect_2012} and graphBior \cite{narang_compact_2013}.
GraphQMF is an orthogonal graph filter bank, which requires eigen-decomposition to achieve exact PR. This is because the polynomial approximation of filters results in reconstruction errors. In contrast, GraphBior is a PR graph filter bank that utilizes polynomial graph filters and can be implemented without requiring eigen-decomposition of the graph operator.

We apply the Harary's decomposition algorithm \cite{harary_biparticity_1977} to those BGFBs for the graph bipartition.
We also perform reconstruction with the single channel sampling as a benchmark. For all methods, we implement analysis and synthesis graph filters with the $50$th order polynomial approximation \cite{shuman_chebyshev_2011}.

\begin{figure}[t!]
\centering
 \subfigure[][Original]
  {\centering\includegraphics[width=0.47\linewidth]{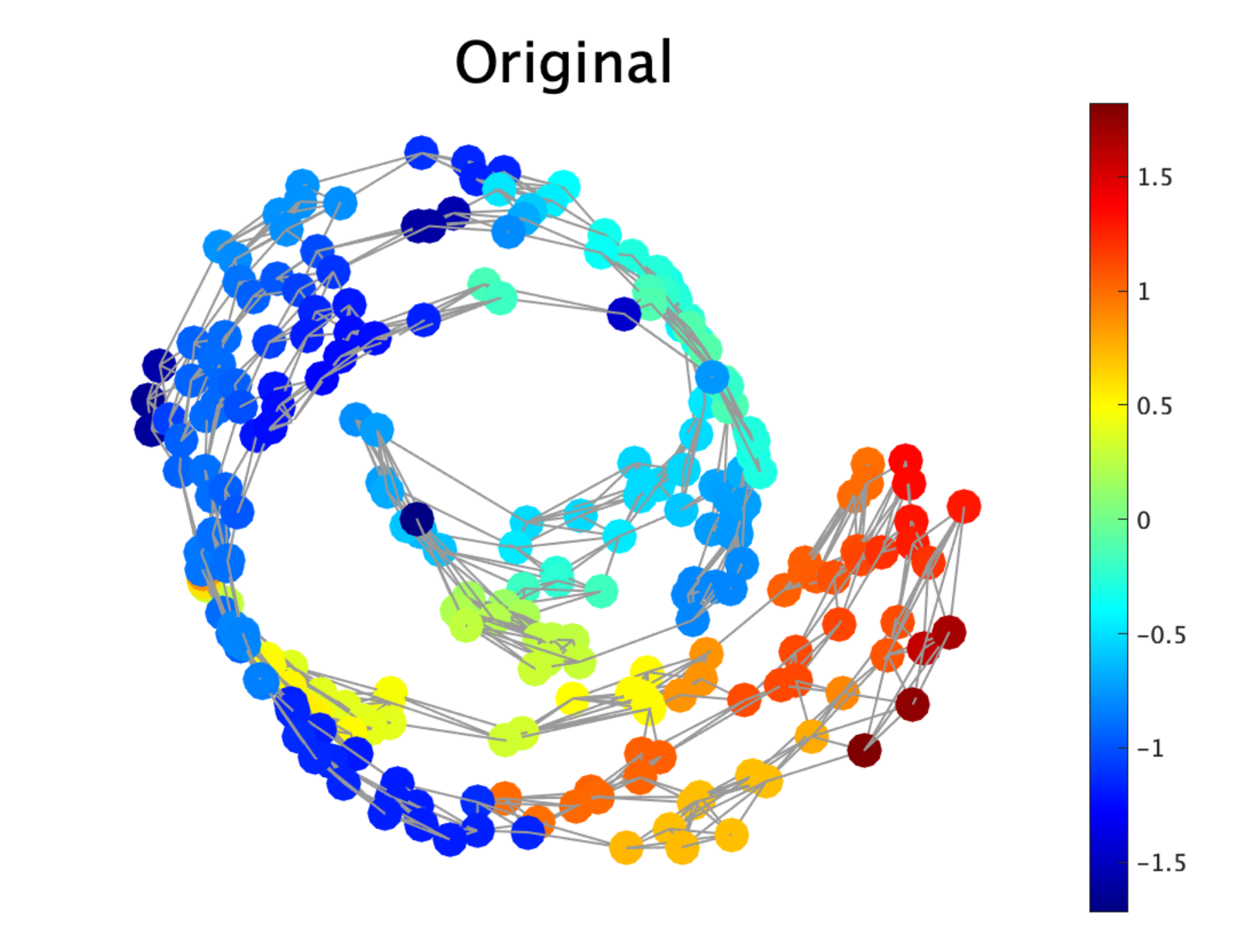} \label{mbp_ori_com}}
 \subfigure[][Proposed]
  {\centering\includegraphics[width=0.47\linewidth]{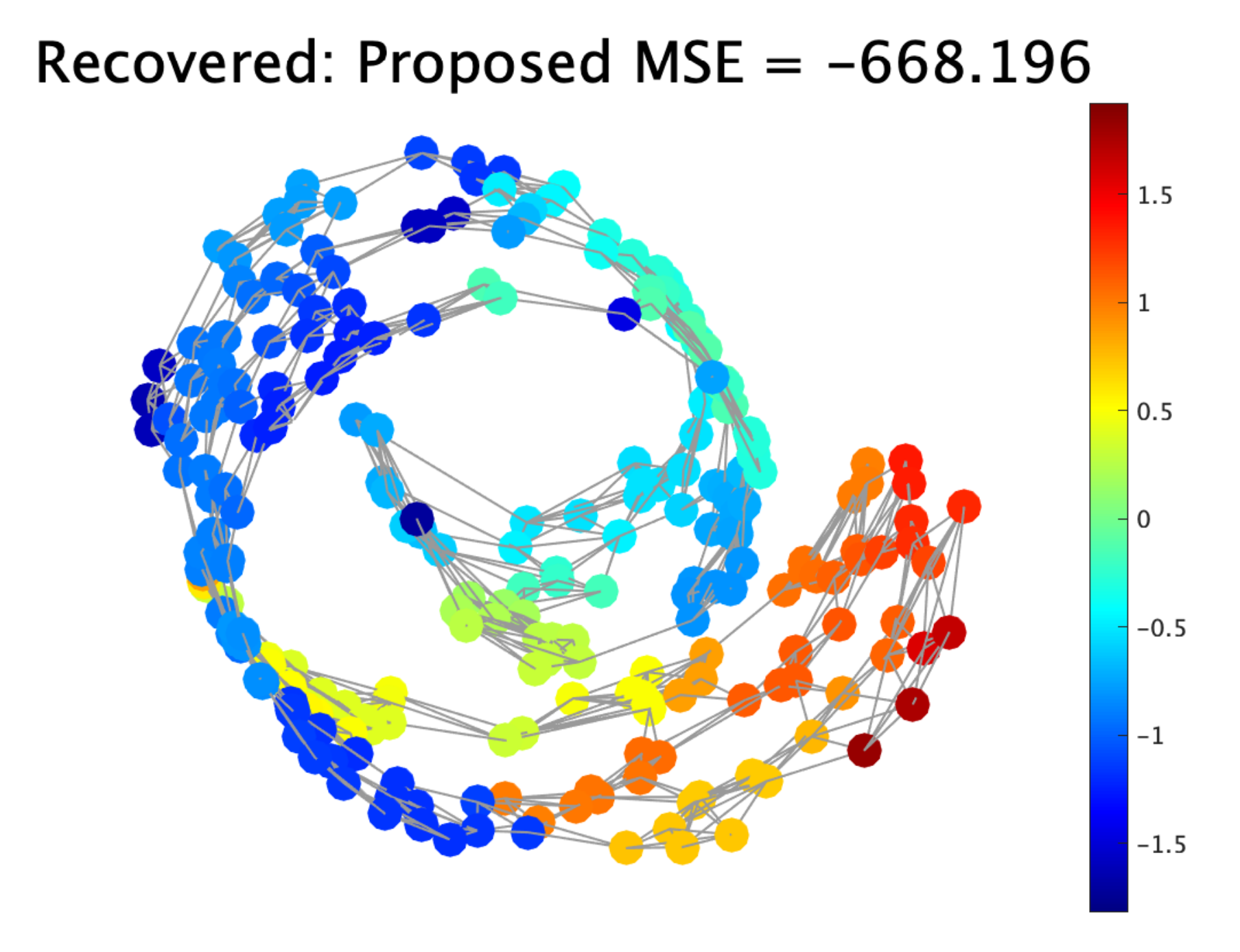} \label{pws_prop_com}}
 \subfigure[][Recon. w/ ch. 1]
  {\centering\includegraphics[width=0.47\linewidth]{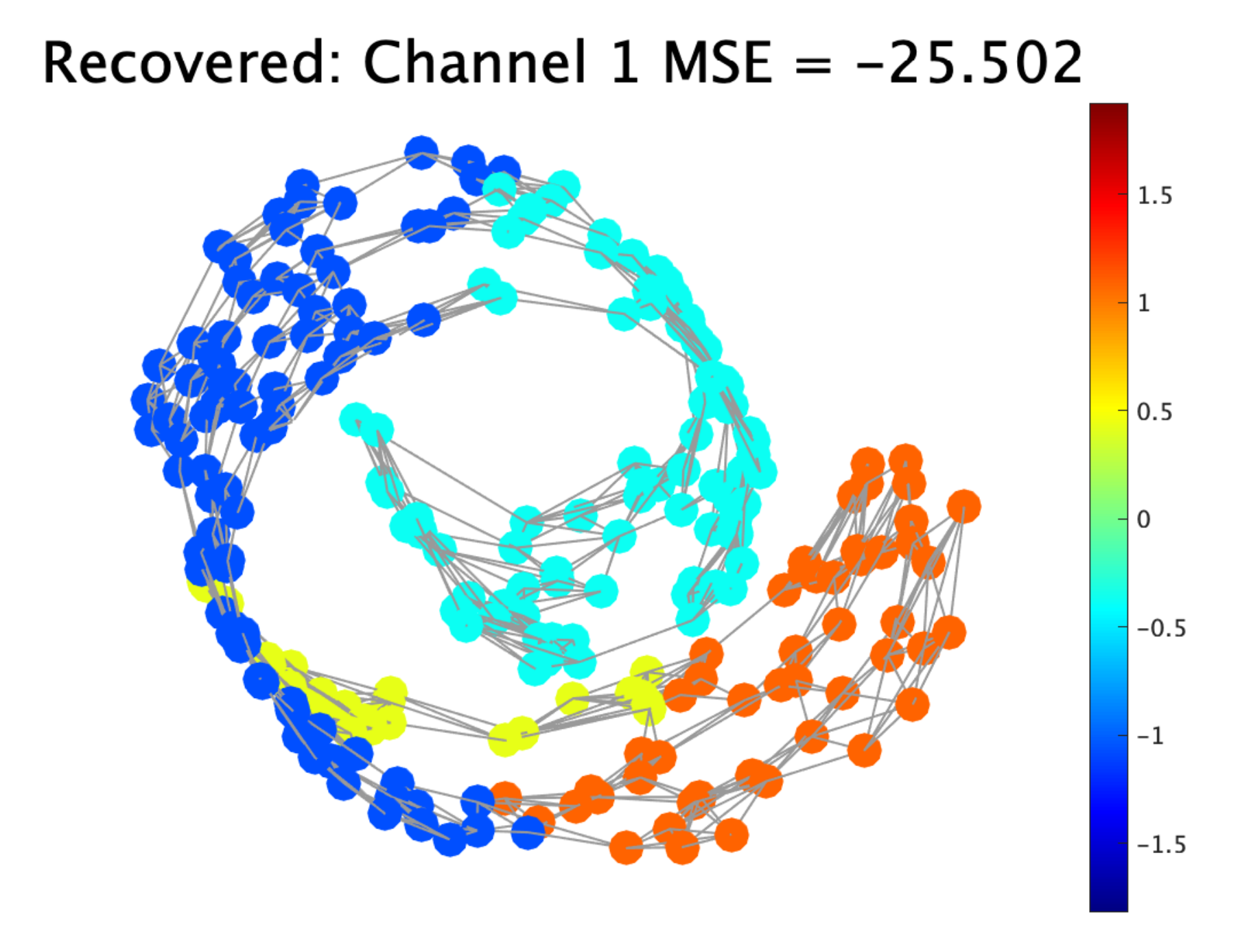} \label{pws_ch1_com}}
 \subfigure[][Recon. w/ ch. 2]
  {\centering\includegraphics[width=0.47\linewidth]{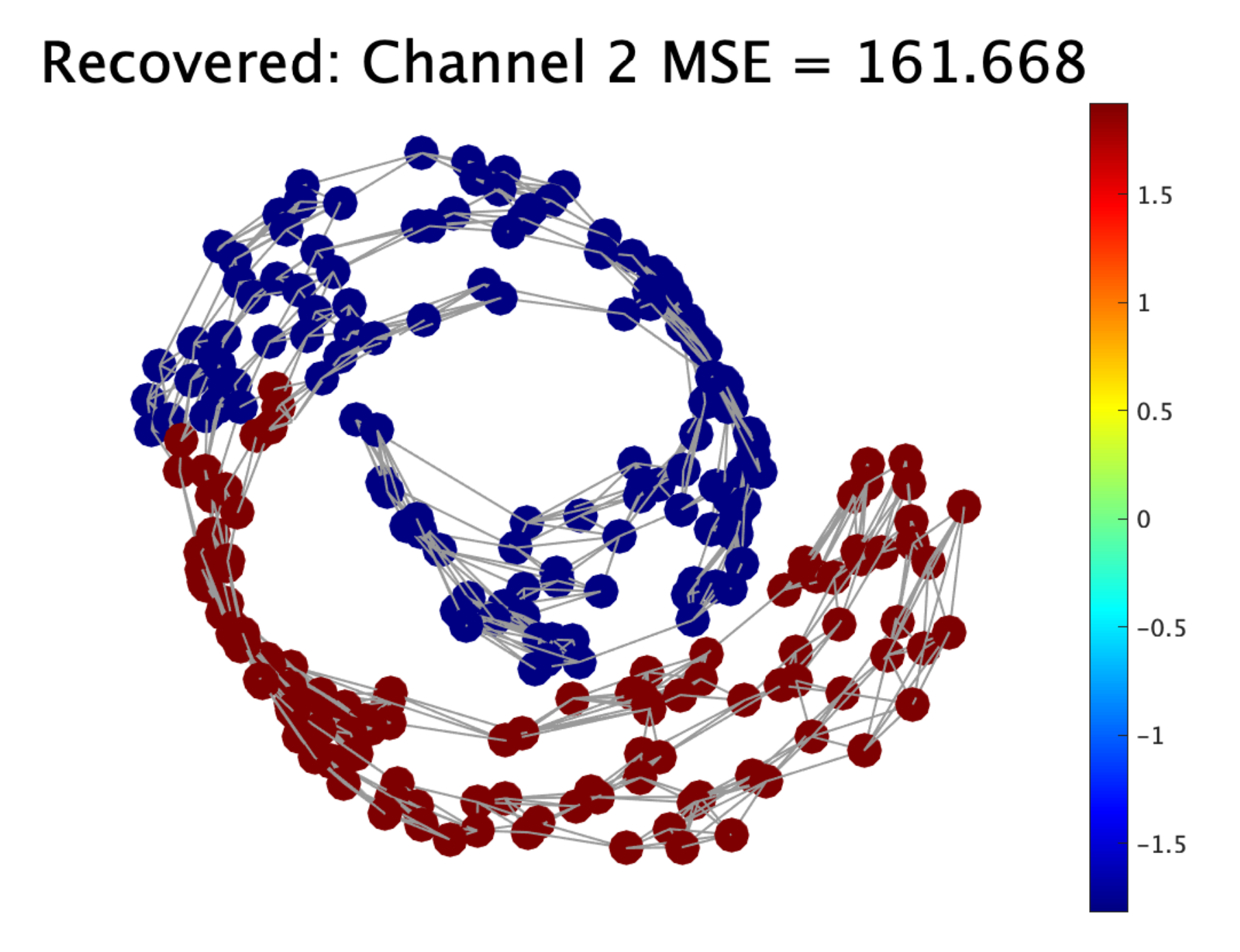} \label{pws_ch2_com}}
 \subfigure[][GraphQMF]
  {\centering\includegraphics[width=0.47\linewidth]{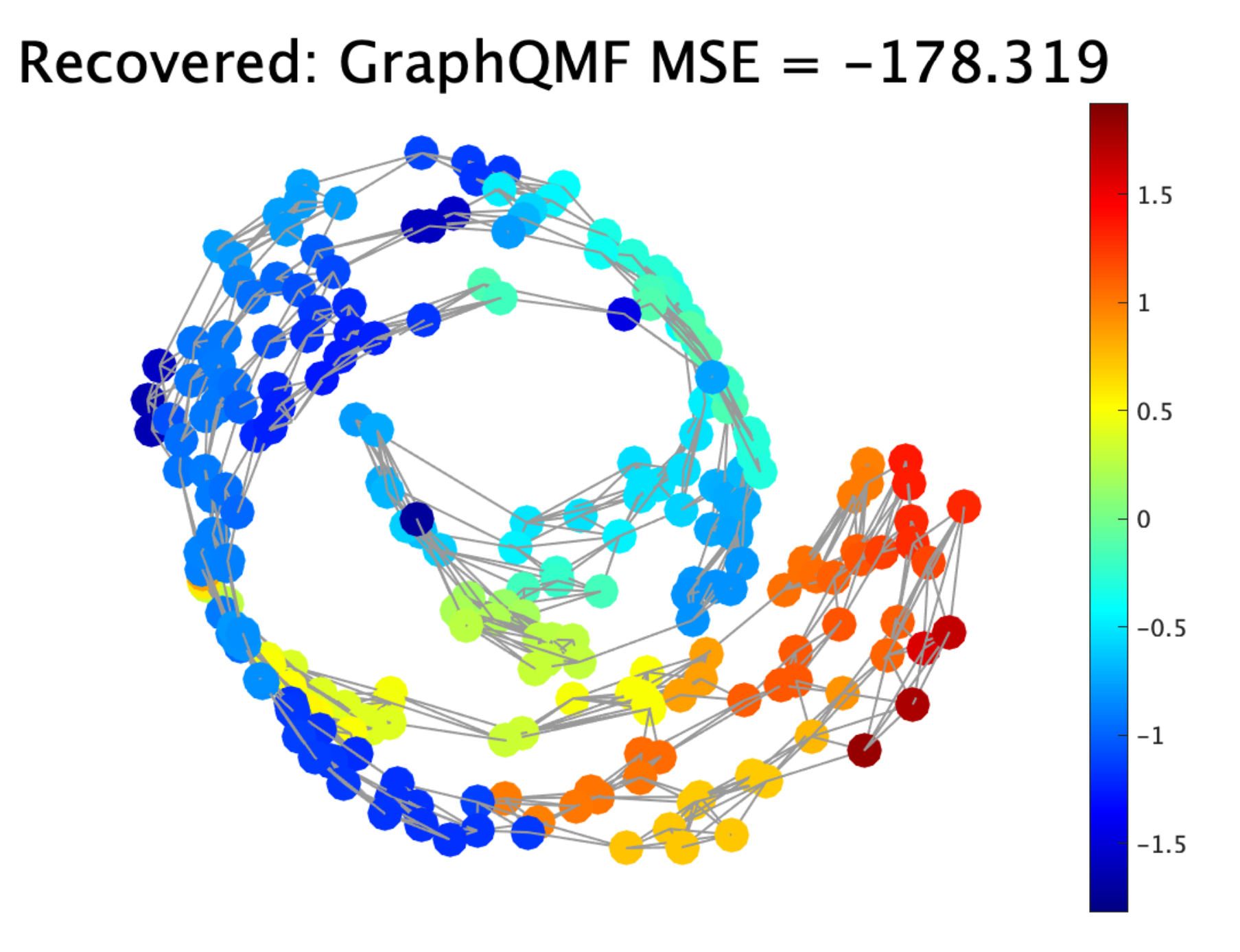} \label{pws_graphQMF_com}}
 \subfigure[][GraphBior]
  {\centering\includegraphics[width=0.47\linewidth]{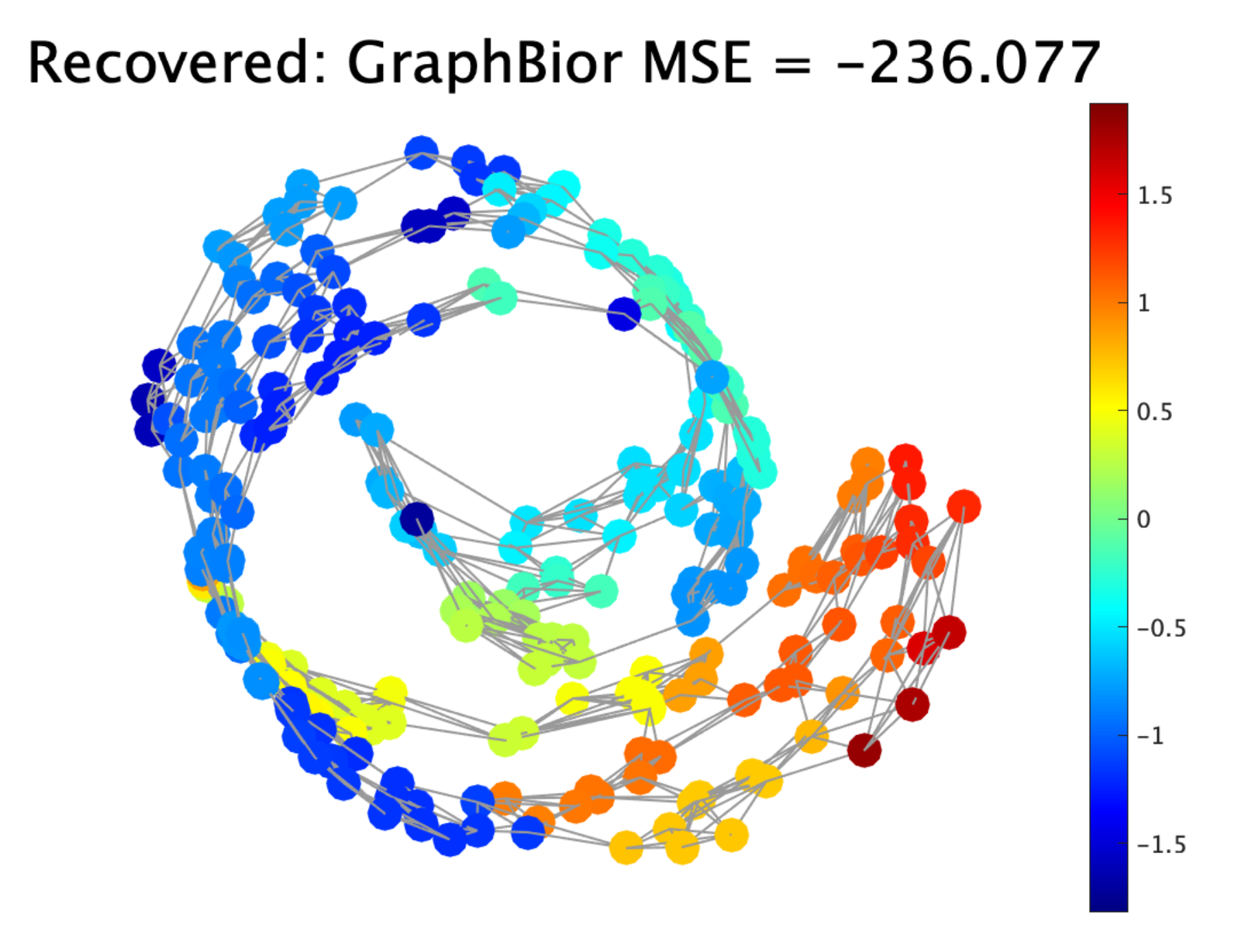} \label{pws_graphBior_com}}
\caption{Examples of recovery for PWS graph signals on swiss roll graphs.}
\label{exp:multi_recov_pws_com}
\end{figure}

\begin{figure}[t!]
\centering
 \subfigure[][Original]
  {\centering\includegraphics[width=0.47\linewidth]{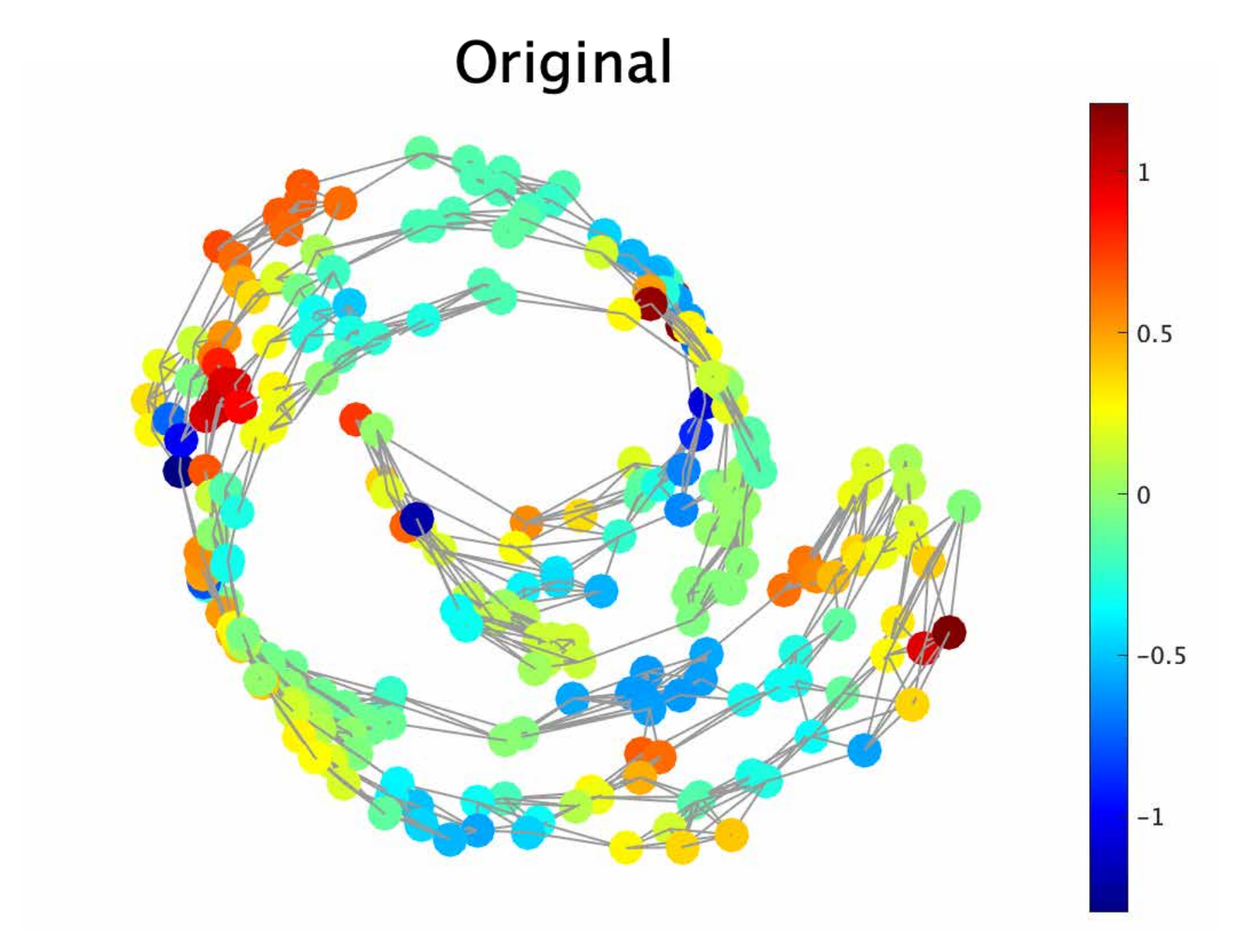} \label{ubp_ori_com}}
 \subfigure[][Proposed]
  {\centering\includegraphics[width=0.47\linewidth]{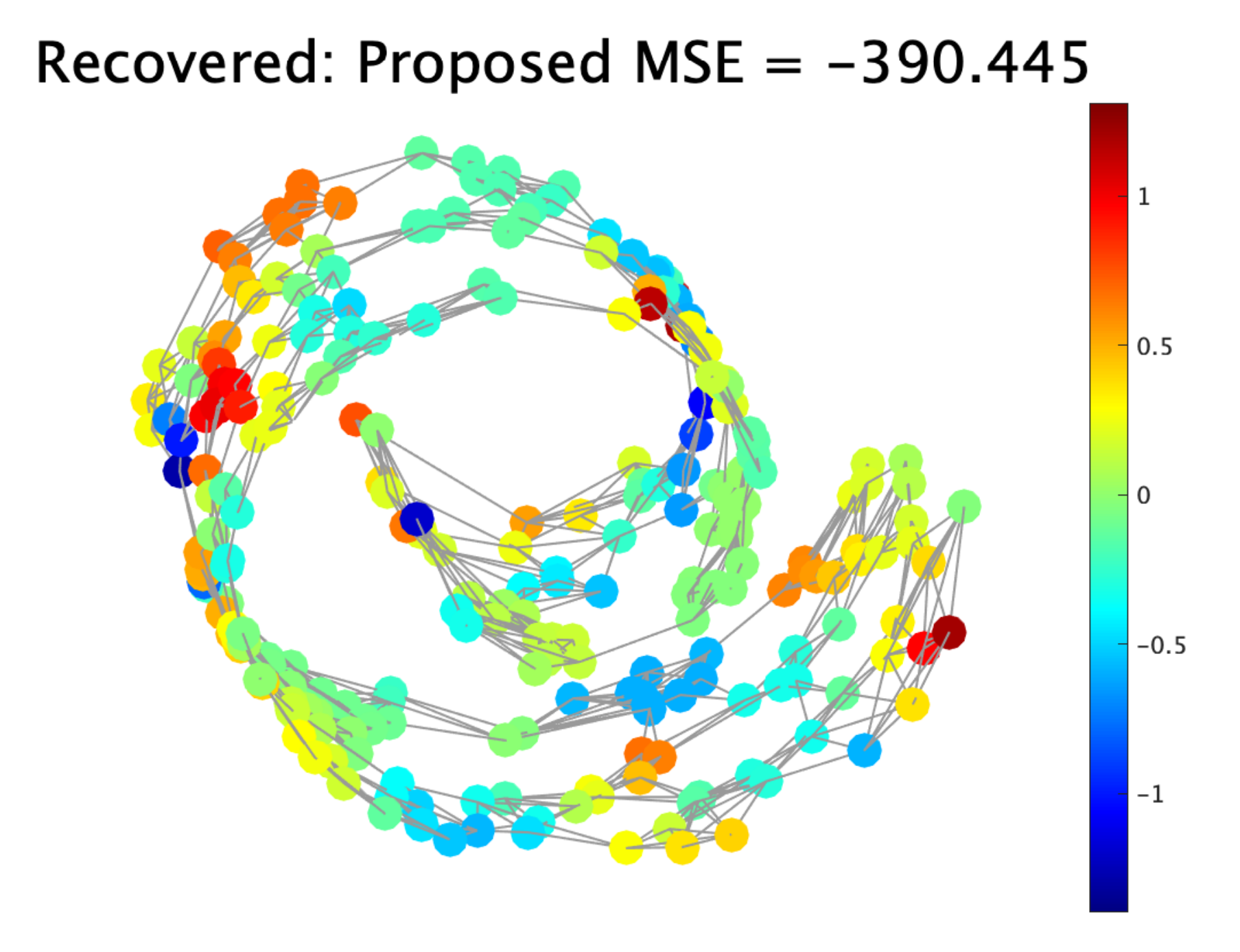} \label{ubp_prop_com}}
 \subfigure[][Recon. w/ ch. 1]
  {\centering\includegraphics[width=0.47\linewidth]{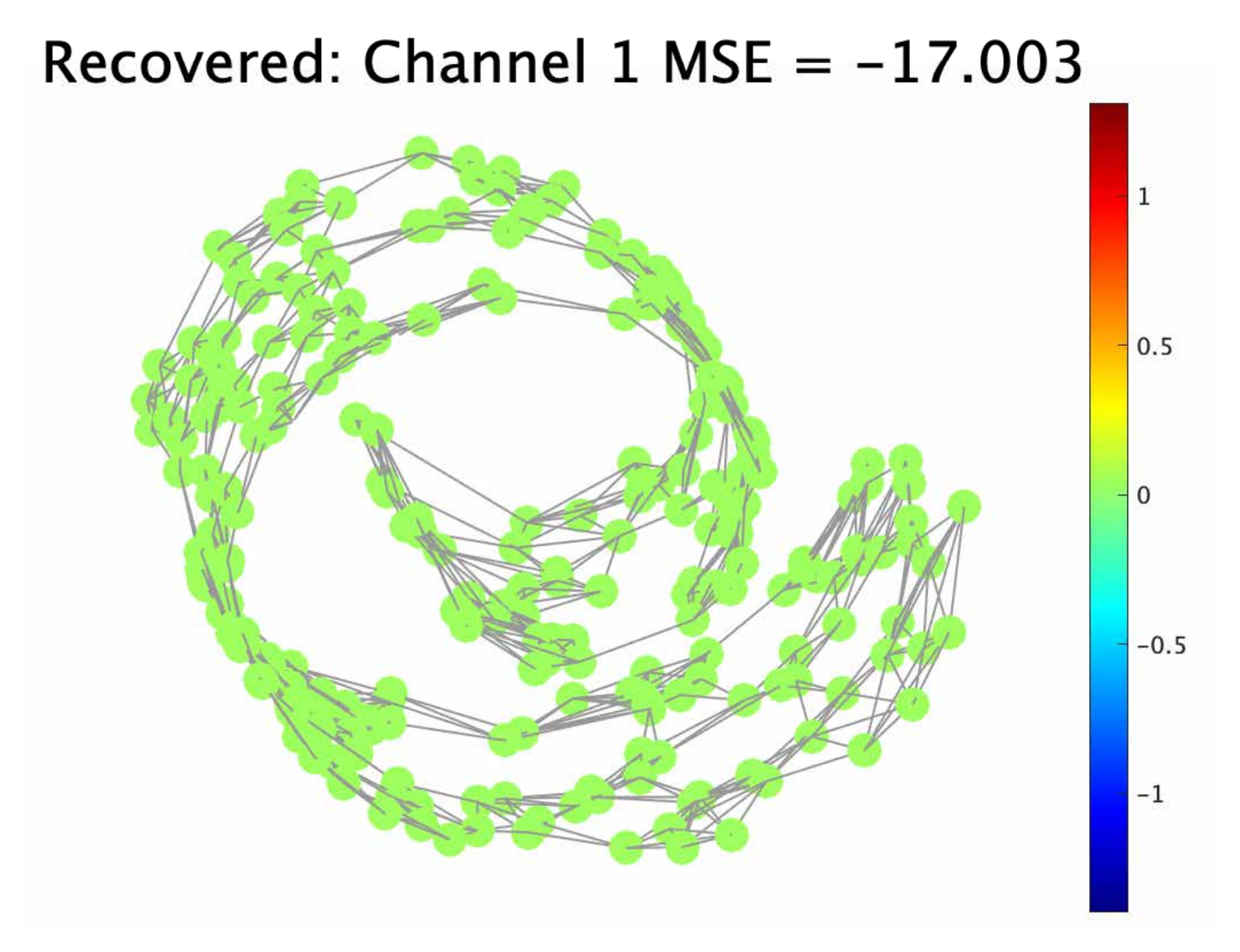} \label{ubp_ch1_com}}
 \subfigure[][Recon. w/ ch. 2]
  {\centering\includegraphics[width=0.47\linewidth]{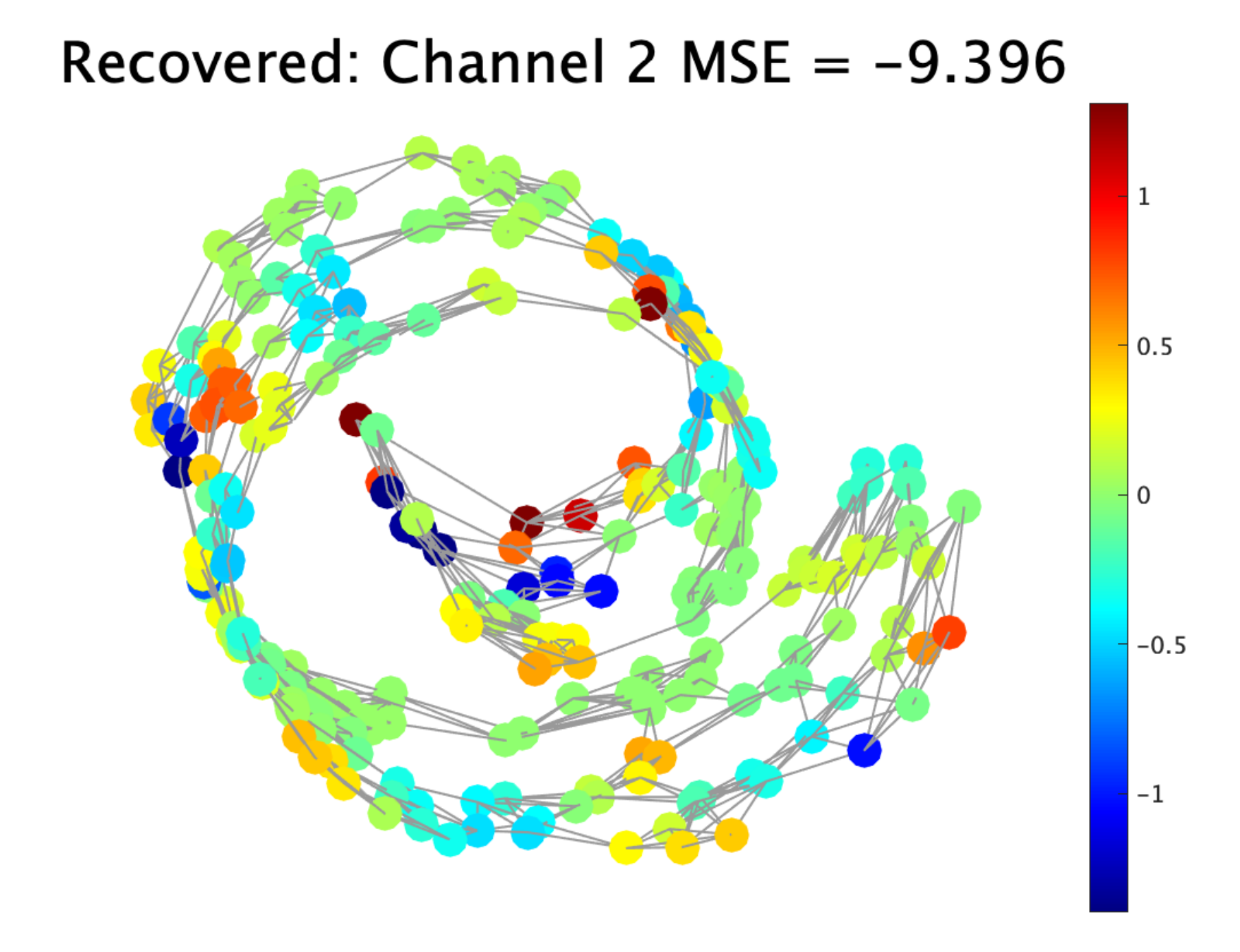} \label{ubp_ch2_com}}
 \subfigure[][GraphQMF]
  {\centering\includegraphics[width=0.47\linewidth]{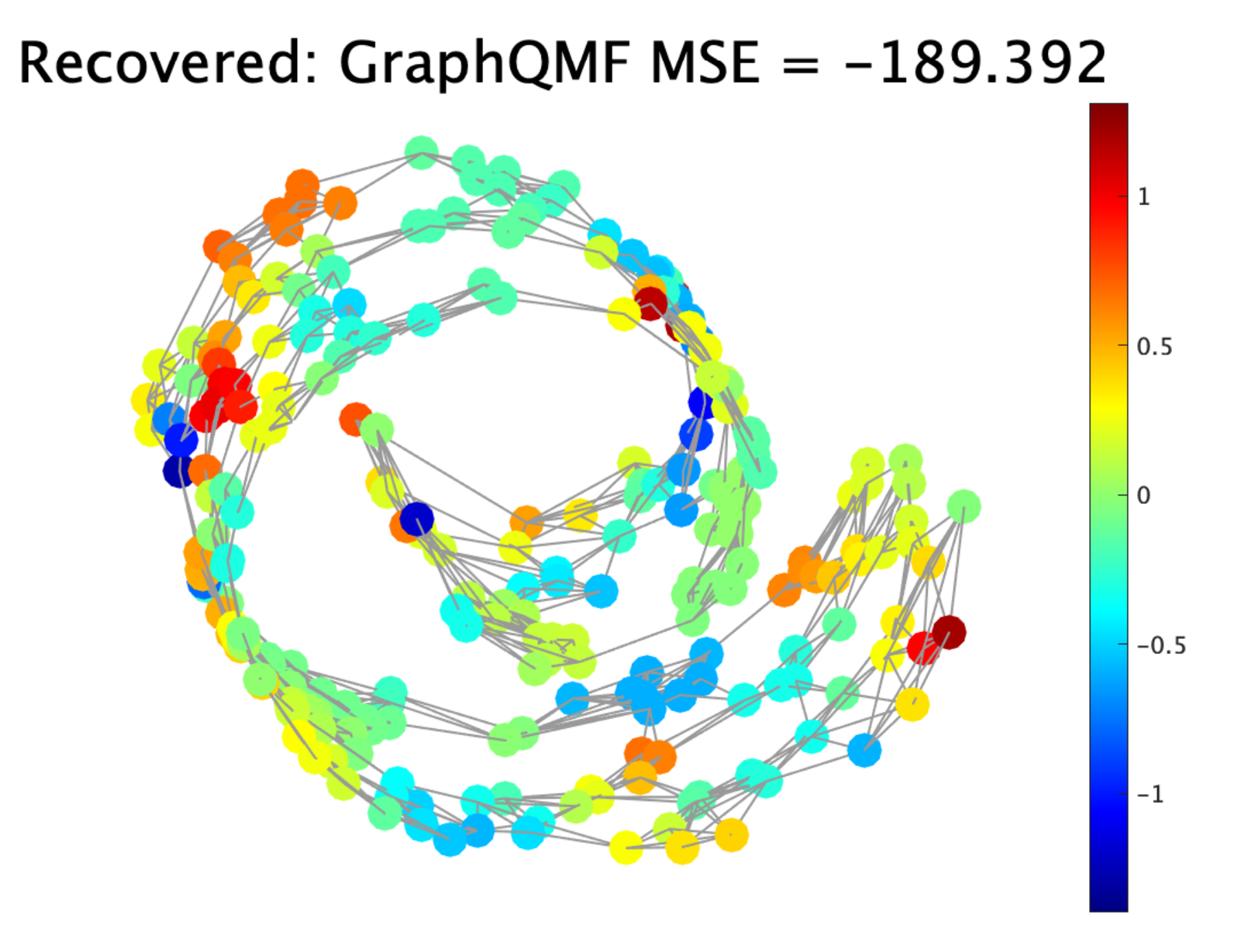} \label{ubp_graphQMF_com}}
 \subfigure[][GraphBior]
  {\centering\includegraphics[width=0.47\linewidth]{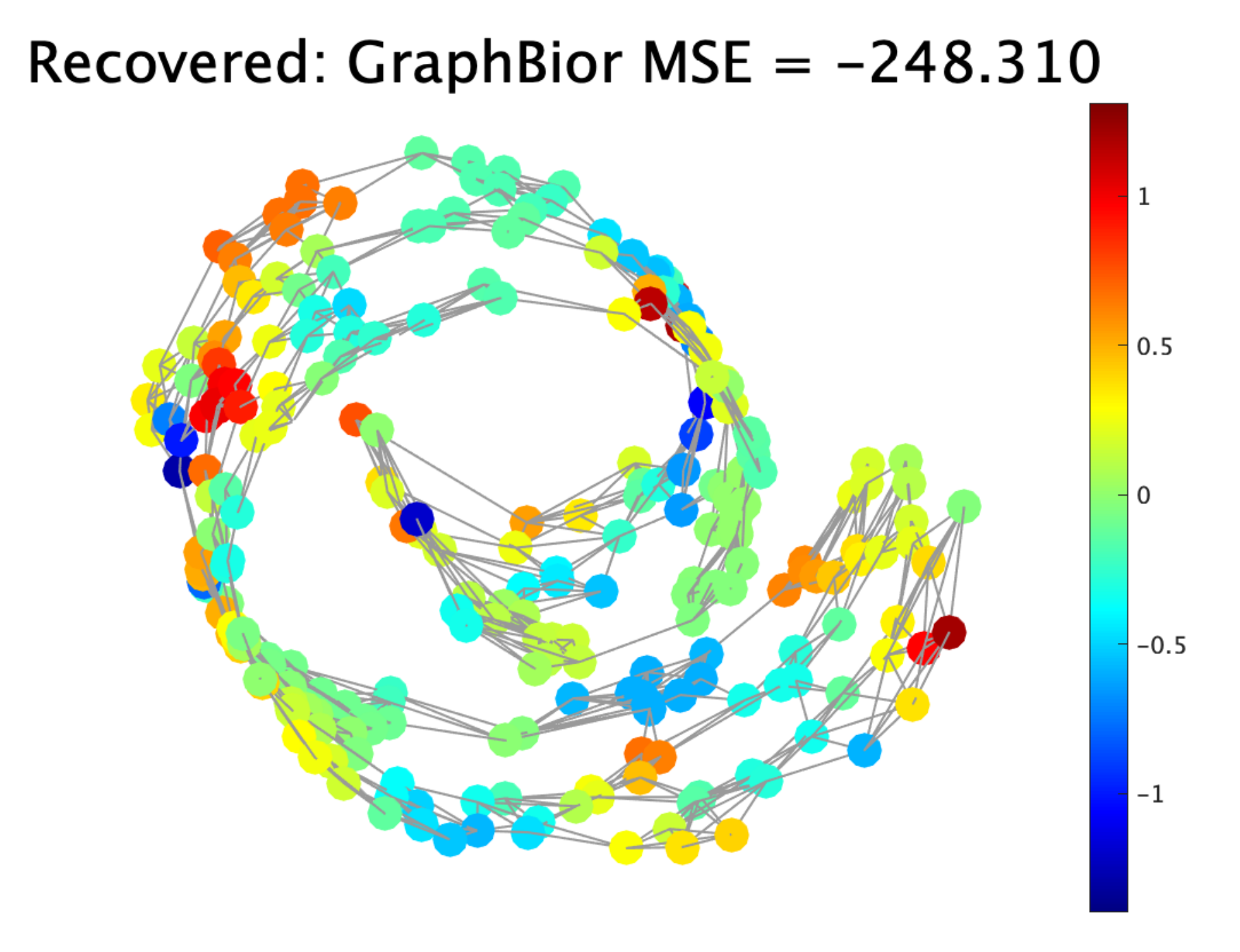} \label{ubp_graphBior_com}}
\caption{Examples of recovery for UBP graph signals on swiss roll graphs.}
\label{exp:multi_recov_ubp_com}
\end{figure}

\subsubsection{RESULTS}
The reconstruction MSEs in decibels are summarized in Table \ref{tab:ave_mse} and examples of the reconstructed graph signals are visualized in Figs. \ref{exp:multi_recov_pws}-\ref{exp:multi_recov_ubp_com}.
As observed in Table \ref{tab:ave_mse}, the proposed MCS best recovers graph signals for both signal models and both graphs. The single-channel sampling failed to recover. That is, the conventional single-channel sampling does not work well for a mixture of multiple graph signal models.
GraphQMF presents the low MSE but it involves slight errors caused by polynomial approximations of graph filters. 
GraphBior and the proposed method can be regarded as PR in machine precision
but our method presents the lowest MSE for all methods regardless of the graphs.

Note that our MCS can achieve PR with a user-specified graph operator, while the other methods require to the use of the normalized graph Laplacian or graph simplification prior to sampling.

\subsection{REAL-WORLD GRAPH}
We then perform a sampling experiment for a real-world graph.

\subsubsection{SETUP}

We utilize a traffic network dataset from the Caltrans Performance Measurement System\footnote{This dataset is publicly available http://pems.dot.ca.gov.}. Vertices represent stations of 17 highways in Alameda county, CA, where $N=593$. Edges connect the vertices if the stations are adjacent on the same highway or if there is a junction close to the stations on different highways. Hereafter, we refer to this graph as Alameda graph.

We consider a synthetic signal on Alameda graph to objectively measure the reconstruction quality since there is no ground-truth available\footnote{
Note that estimating generator functions (both for single- and multi-channel cases) is an open problem while some recent studies are undergoing \cite{hara_graph_2022,behjat_signal-adapted_2016}.}. In this experiment, we generate PWS graph signals according to \eqref{eq:pws_model} where we divide Alameda graph into three clusters with spectral clustering \cite{von_luxburg_tutorial_2007}.
The sampling ratio is set to $K=|\mathcal{M}_0|=297$. Analysis filters are the same as the previous experiment.

We calculate the average MSE of reconstructed graph signals for 30 independent runs, and compare it with the existing methods from the previous experiment. 

\subsubsection{RESULTS}

The MSEs in decibels are summarized in Table \ref{tab:ave_mse_real} and reconstructed graph signals are visualized in Fig. \ref{exp:multi_recov_pws_alameda}. Similar to the previous experiment, we observe that the proposed MCS exhibits the lowest MSEs among all the methods. This suggests that our MCS can be applied to real-world datasets, regardless of the topology of graphs.

\begin{figure}[t!]
\centering
 \subfigure[][Original]
  {\centering\includegraphics[width=0.47\linewidth]{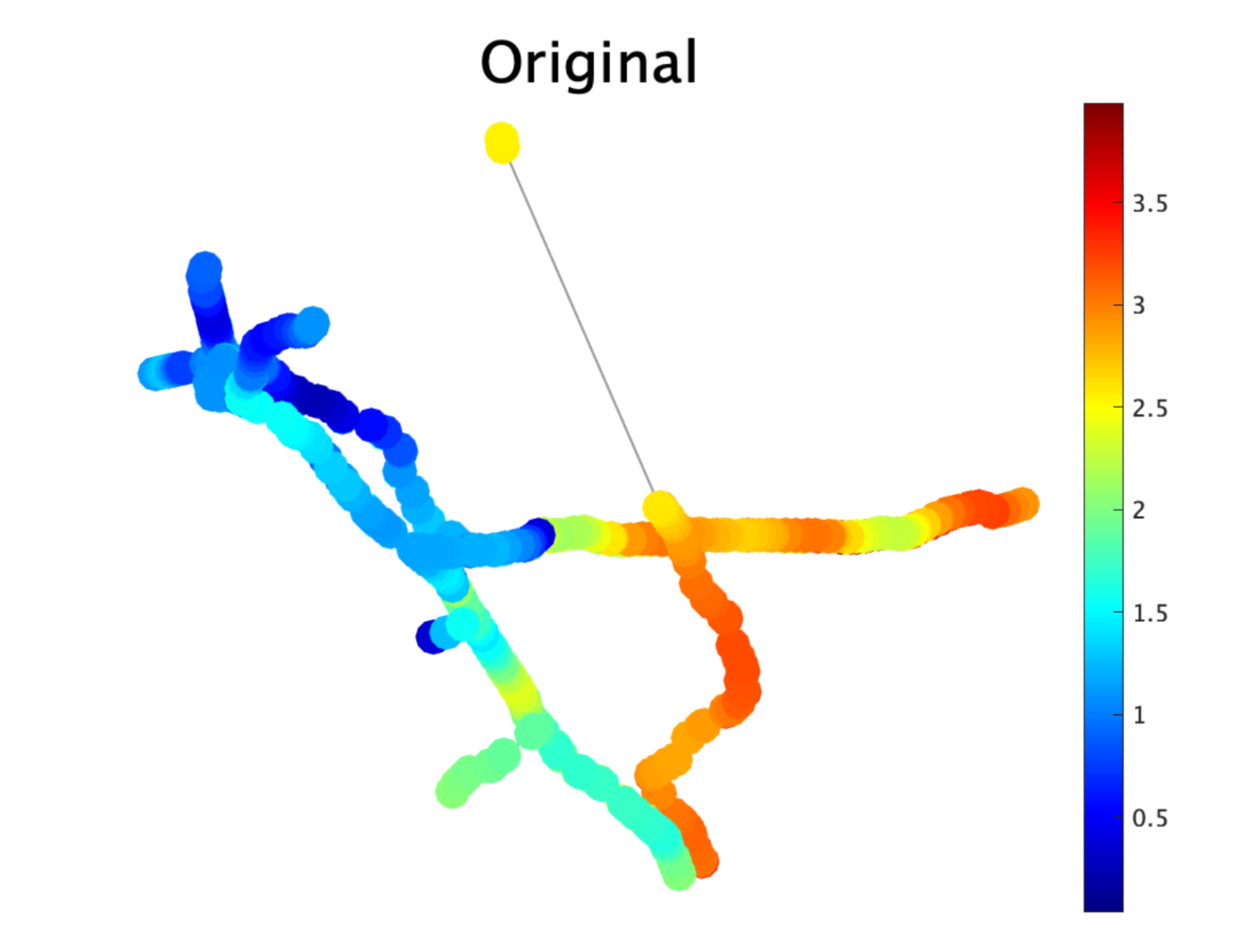} \label{mbp_ori_com}}
 \subfigure[][Proposed]
  {\centering\includegraphics[width=0.47\linewidth]{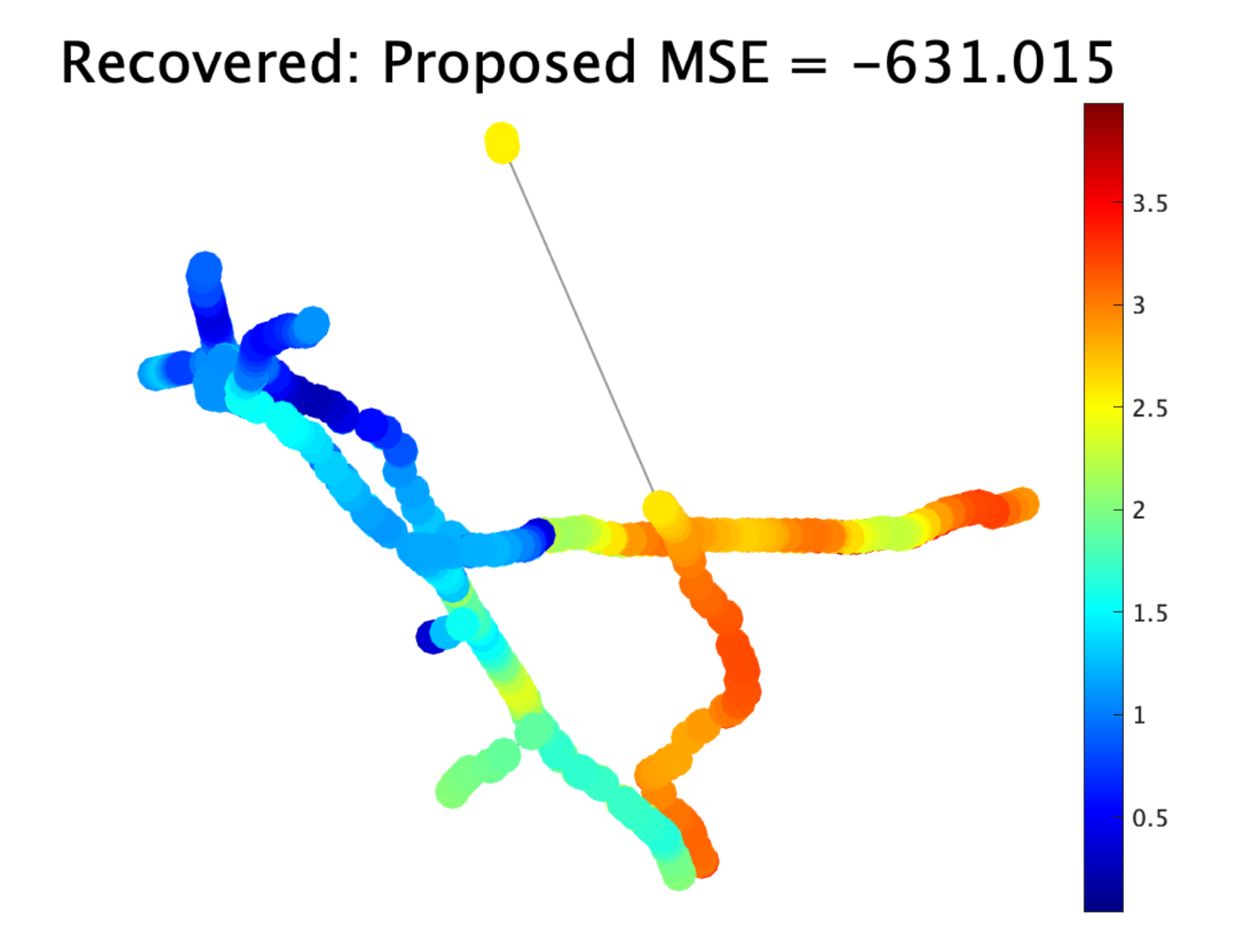} \label{pws_prop_com}}
 \subfigure[][Recon. w/ ch. 1]
  {\centering\includegraphics[width=0.47\linewidth]{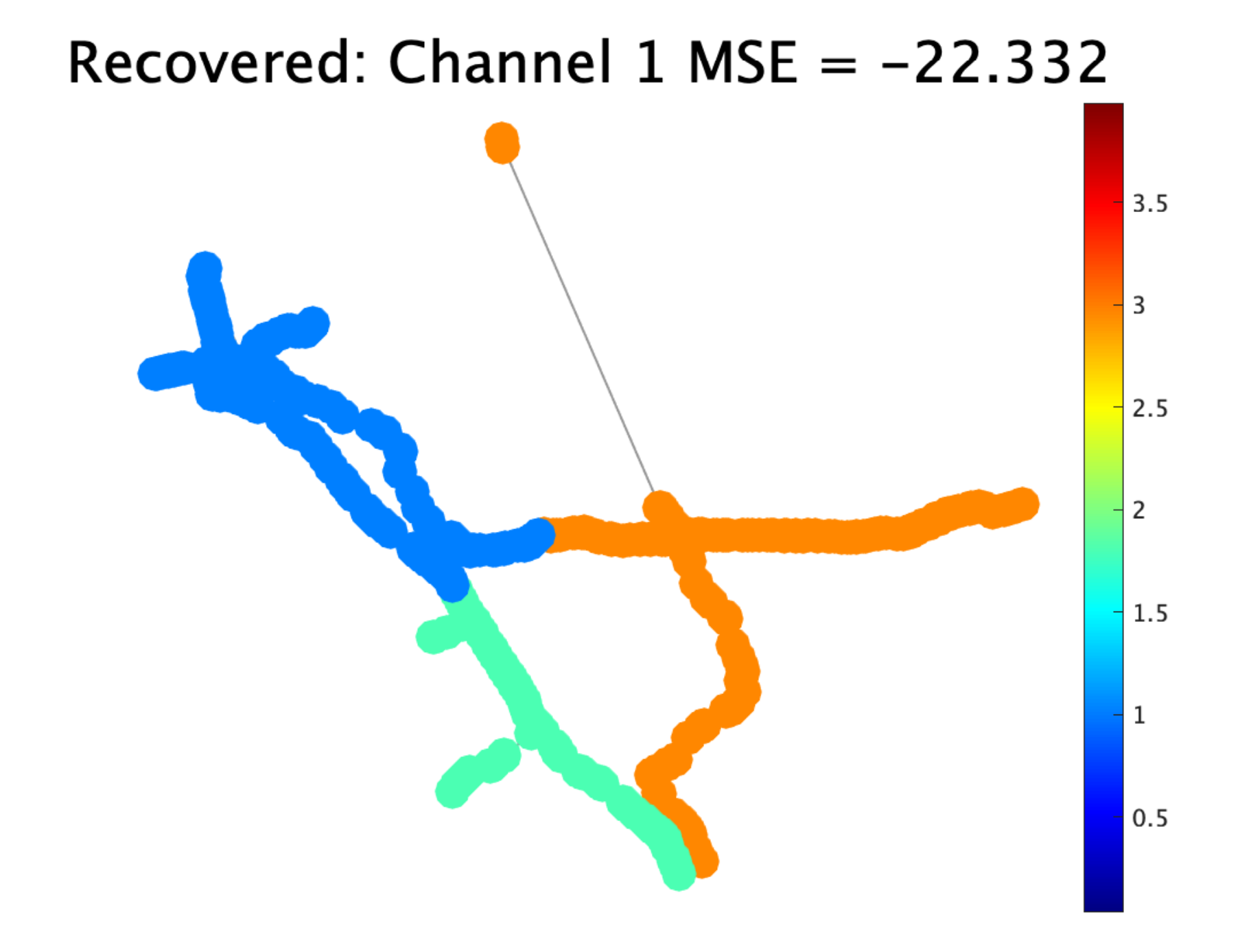} \label{pws_ch1_com}}
 \subfigure[][Recon. w/ ch. 2]
  {\centering\includegraphics[width=0.47\linewidth]{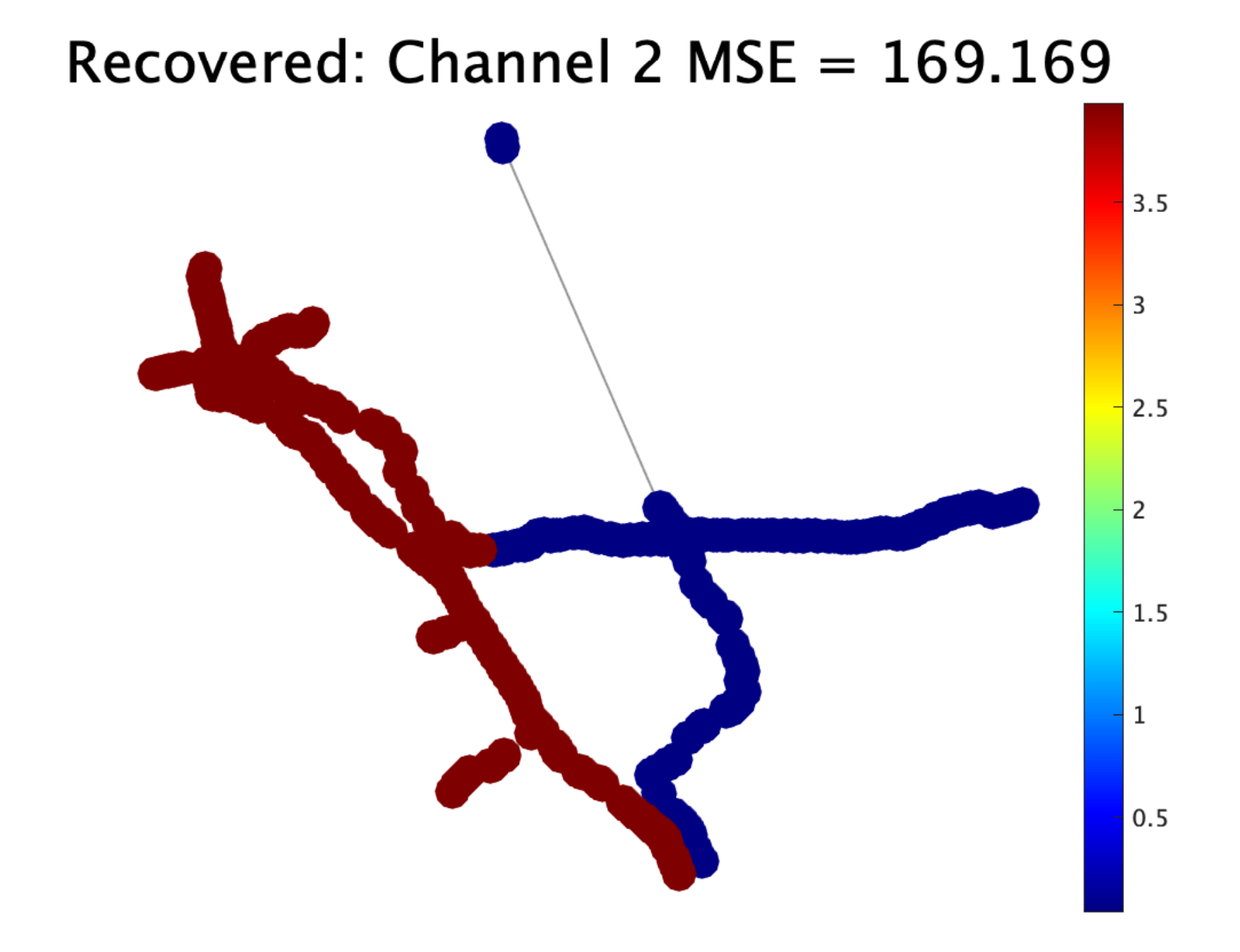} \label{pws_ch2_com}}
 \subfigure[][GraphQMF]
  {\centering\includegraphics[width=0.47\linewidth]{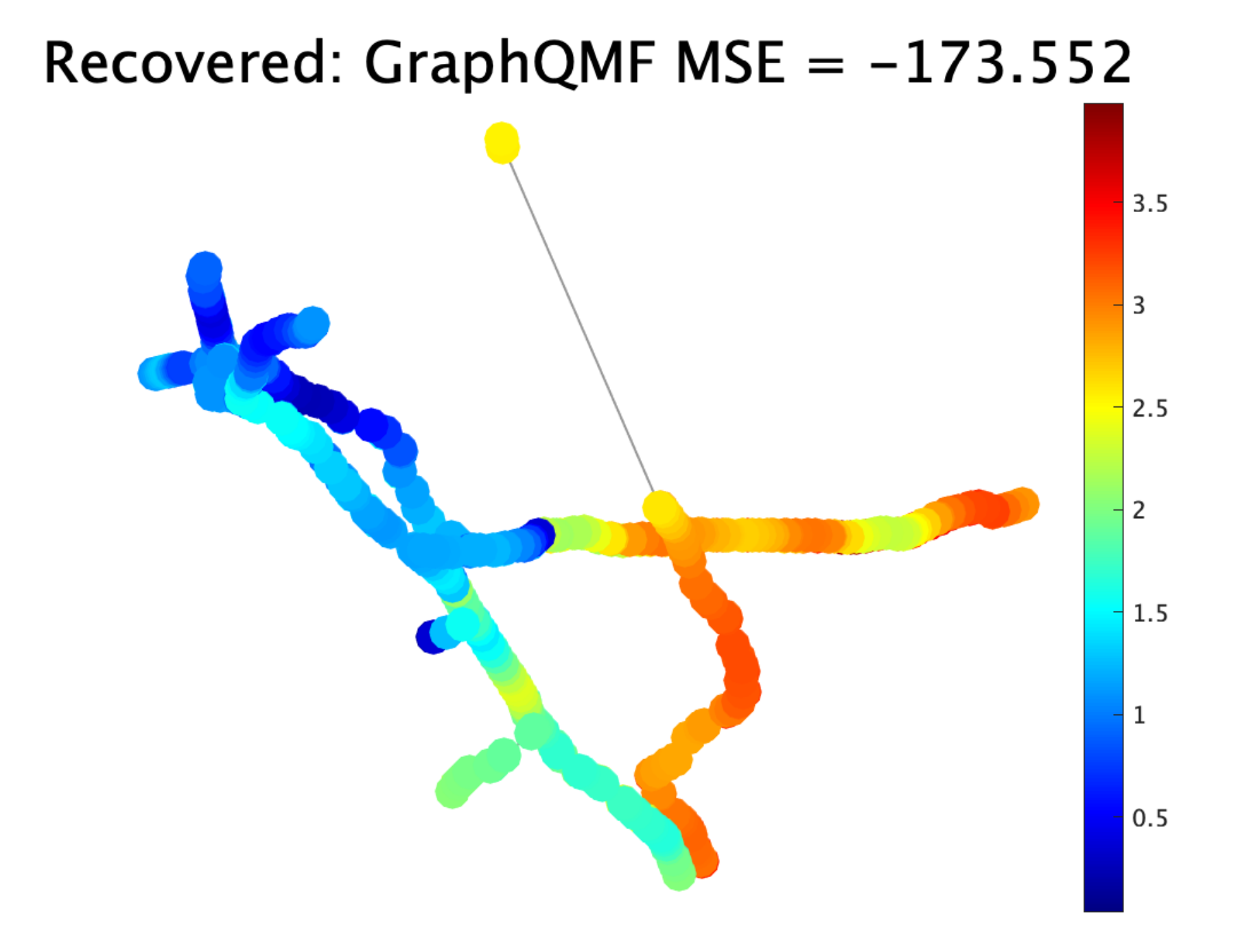} \label{pws_graphQMF_com}}
 \subfigure[][GraphBior]
  {\centering\includegraphics[width=0.47\linewidth]{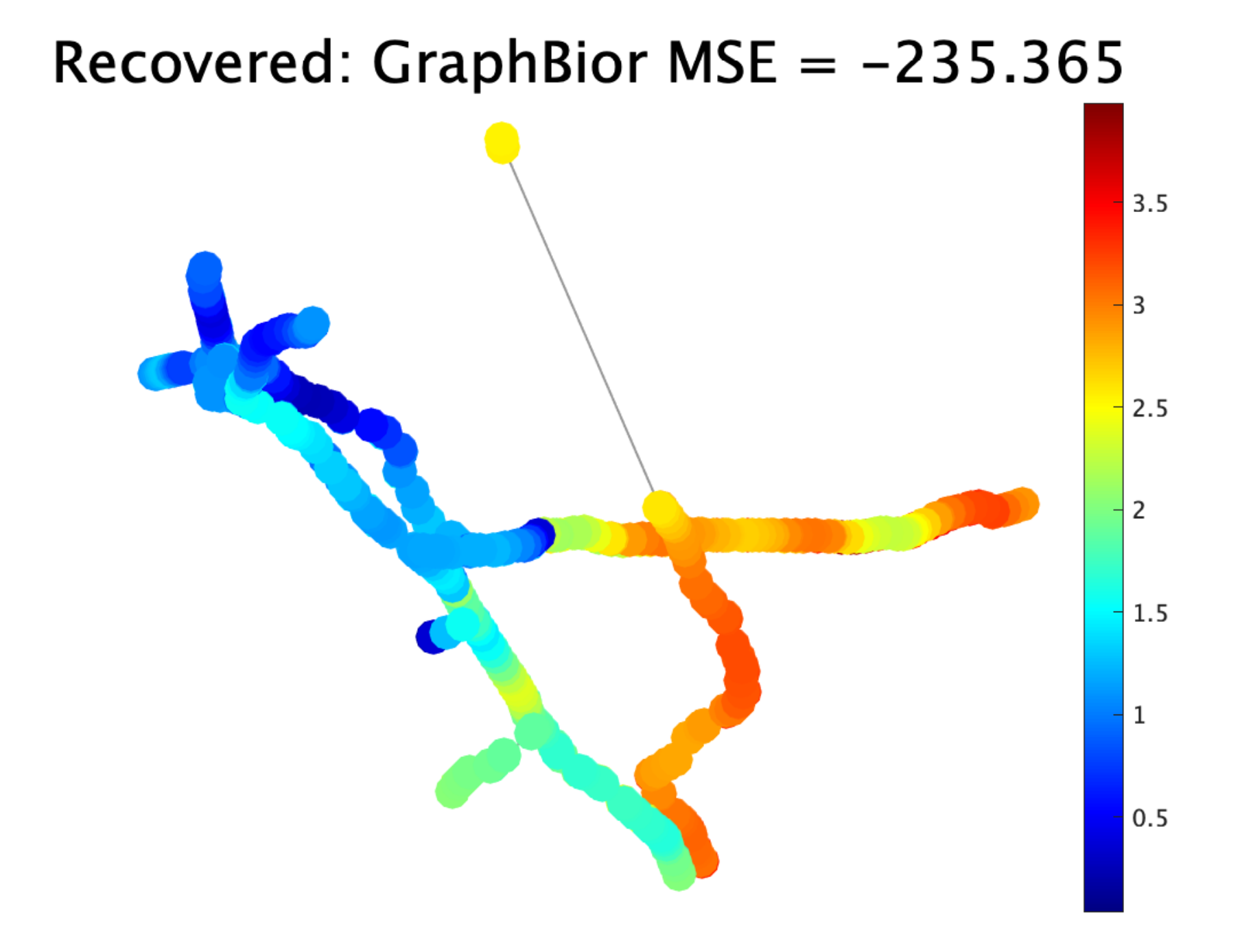} \label{pws_graphBior_com}}
\caption{Examples of recovery for PWS graph signals on Alameda graph.}
\label{exp:multi_recov_pws_alameda}
\end{figure}

\begin{table}[t!]
\centering
\caption{Average MSEs of 30 independent runs for recovery on Alameda graph. }\label{tab:ave_mse_real}
\setlength{\tabcolsep}{0.5em}
\begin{tabular}{ccccc}
\hline
Proposed & \begin{tabular}[c]{@{}c@{}}Recon. \\ w/ ch.~1\end{tabular} & \begin{tabular}[c]{@{}c@{}}Recon. \\ w/ ch.~2\end{tabular} & GraphQMF & GraphBior \\ \hline\hline
\textbf{-633.29} & -20.73 & 196.12 & -172.59 & -234.86 \\\hline
\end{tabular}
\end{table}

\section{CONCLUSION}

In this paper, we develop the first MCS framework for graph signals, by extending the single-channel graph signal sampling to the multi-channel setting. We present a SSS method for our MCS such that graph signals are best recovered. We reveal that existing BGFBs are a special case of the proposed MCS. We demonstrate the effectiveness of the proposed method by showing that our MCS outperforms the existing BGFBs and the single-channel sampling.

\bibliographystyle{IEEEtran}
\bibliography{IEEEabrv,multi_channel_samp_ojsp_cleaned}

\end{document}